\documentclass[11pt]{article}

\usepackage{amsfonts,amsmath,amssymb,amsthm,color,geometry,cite,graphicx}
\usepackage[T1]{fontenc}
\usepackage{selinput}
\SelectInputMappings{
  adieresis={ä},
  germandbls={ß}
}
\usepackage[english]{babel}
\usepackage [normalem]{ulem}

\newcounter{margcount} 
\newtheorem{theorem}{Theorem}[section]
\newtheorem{proposition}[theorem]{Proposition}
\newtheorem{definition}[theorem]{Definition}
\newtheorem{lemma}[theorem]{Lemma}
\newtheorem{corollary}[theorem]{Corollary}
\newtheorem{Remark}[theorem]{Remark}

\newcommand{\SKP}[2]{\langle #1, #2 \rangle_{}}

\newcommand{\SKPLL}[4]{\text{}_{\scriptscriptstyle#3}\langle #1, #2
\rangle_{{\scriptscriptstyle {#4}}}}
\newcommand{\NNN}[2]{\|#1\|_{#2}}

\newcommand{\diam}{{\rm diam \ }}

\newcounter{pcounter} 
 
\numberwithin{equation}{section}
\numberwithin{theorem}{section}

\newcommand{\alp}{\alpha}

\newcommand{\eps}{\varepsilon}

\newcommand{\lam}{\lambda}
\renewcommand{\phi}{\varphi}

\newcommand{\Gam}{\Gamma}

\renewcommand{\t}{\tilde}
\newcommand{\p}{\partial}

\newcommand{\nco}[1]{\|#1\|_{L^\infty(\TTT)}}

\def \wto{\rightharpoonup}

\DeclareMathOperator*{\supp}{supp}

\renewcommand{\AA}{{\mathcal A}}

\newcommand{\DD}{{\mathcal D}}

\newcommand{\HH}{{\mathcal H}}
\newcommand{\II}{{\mathcal I}}

\newcommand{\MM}{{\mathcal M}}

\renewcommand{\phi}{\varphi}

\renewcommand{\H}{\mathcal{H}}

\newcommand{\I}{\mathcal{I}}

\newcommand{\nada}[1]{}

\newcommand{\Rn}{{\mathbb R^3}}
\newcommand{\R}{{\mathbb R}}

\newcommand{\N}{{\mathbb N}}

\newcommand{\TTT}{{\mathbb T}}

\newcommand{\torol}{{\mathbb T_\ell}}

\definecolor{light}{gray}{.97}

\title{Low density phases in a uniformly charged liquid 
}

\author{Hans Kn\"upfer\footnote{Institut f\"ur Angewandte Mathematik,
    Universit\"at Heidelberg, INF 294, 69120 Heidelberg, Germany} \and
  Cyrill B. Muratov \footnote{Department of Mathematical Sciences, New
    Jersey Institute of Technology, Newark, NJ 07102, USA} \and Matteo
  Novaga \footnote{Dipartimento di Matematica, Universit\`a di Pisa,
    Largo Bruno Pontecorvo 5, 56127 Pisa, Italy}}

\date{\today}

\begin{document}

\maketitle

\begin{abstract}  
  This paper is concerned with the macroscopic behavior of global
  energy minimizers in the three-dimensional sharp interface
  unscreened Ohta-Kawasaki model of diblock copolymer melts. This
  model is also referred to as the nuclear liquid drop model in the
  studies of the structure of highly compressed nuclear matter found
  in the crust of neutron stars, and, more broadly, is a paradigm for
  energy-driven pattern forming systems in which spatial order arises
  as a result of the competition of short-range attractive and
  long-range repulsive forces. Here we investigate the large volume
  behavior of minimizers in the low volume fraction regime, in which
  one expects the formation of a periodic lattice of small droplets of
  the minority phase in a sea of the majority phase.  Under periodic
  boundary conditions, we prove that the considered energy
  $\Gamma$-converges to an energy functional of the limit
  ``homogenized'' measure associated with the minority phase
  consisting of a local linear term and a non-local quadratic term
  mediated by the Coulomb kernel. As a consequence, asymptotically the
  mass of the minority phase in a minimizer spreads uniformly across
  the domain. Similarly, the energy spreads uniformly across the
  domain as well, with the limit energy density minimizing the energy
  of a single droplets per unit volume. Finally, we prove that in the
  macroscopic limit the connected components of the minimizers have
  volumes and diameters that are bounded above and below by universal
  constants, and that most of them converge to the minimizers of the
  energy divided by volume for the whole space problem.
\end{abstract}

\newpage

\tableofcontents

\section{Introduction} \label{sec-intro} %

The liquid drop model of the atomic nucleus, introduced by Gamow in
1928, is a classical example of a model that gives rise to a geometric
variational problem characterized by a competition of short-range
attractive and long-range repulsive forces
\cite{gamow30,weizsacker35,bohr36,bohr39} (for more recent studies,
see e.g.  \cite{cohen62,myers66,cohen74,pelekasis90,myers96}; for a
recent non-technical overview of nuclear models, see, e.g.,
\cite{cook}). In a nucleus, different nucleons attract each other via
the short-range nuclear force, which, however, is counteracted by the
long-range Coulomb repulsion of the constitutive protons. Within the
liquid drop model, the effect of the short-range attractive forces is
captured by postulating that the nucleons form an incompressible fluid
with fixed nuclear density and by penalizing the interface between the
nuclear fluid and vacuum via an effective surface tension. The effect
of Coulomb repulsion is captured by treating the nuclear charge as
uniformly spread throughout the nucleus. A competition of the cohesive
forces which try to minimize the interfacial area of the nucleus and
the repulsive Coulomb forces that try to spread the charges apart
makes the nucleus unstable at sufficiently large atomic numbers,
resulting in nuclear fission
\cite{meitner39,bohr39,feenberg39,frenkel39}.

It is worth noting that the liquid drop model is also applicable to
systems of many strongly interacting nuclei. Such a situation arises
in the case of matter at very high densities, occurring, for example,
in the core of a white dwarf star or in the crust of a neutron star,
where large numbers of nucleons are confined to relatively small
regions of space by gravitational forces
\cite{baym71,koester90,pethick95}. As was pointed out independently by
Kirzhnits, Abrikosov and Salpeter, at sufficiently low temperatures
and not too high densities compressed matter should exhibit
crystallization of nuclei into a body-centered cubic crystal in a sea
of delocalized degenerate electrons
\cite{kirzhnits60,abrikosov61,salpeter61}. At yet higher densities,
more exotic nuclear ``pasta phases'' are expected to appear as a
consequence of the effect of ``neutron drip''
\cite{ravenhall83,hashimoto84,oyamatsu84,lattimer85,
  lorenz93,pethick95,okamoto13,schneider13} (for an illustration, see
Fig. \ref{f:pasta}). In all cases, the ground state of nuclear matter
is determined by minimizing the appropriate (free) energy per unit
volume of one of the phases that contains contributions from the
interface area and the Coulomb energy of the nuclei.

\begin{figure}
  \centering
  \includegraphics[width=12cm]{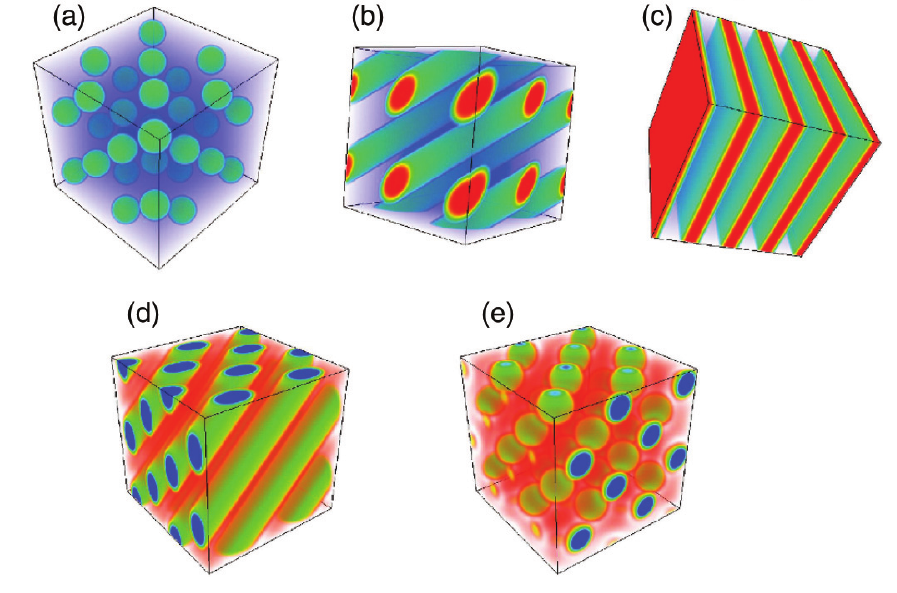}
  \caption{Nuclear pasta phases in a relativistic mean-field model of
    low density nuclear matter. The panels show a progression from
    ``meatball'' (a) to ``spaghetti'' (b) to ``lasagna'' (c) to
    ``macaroni'' (d) to ``swiss cheese'' (e) phases, which are the
    numerically obtained candidates for the ground state at different
    nuclear densities. Reproduced from Ref. \cite{okamoto13}.}
  \label{f:pasta}
\end{figure}

Within the liquid drop model, the simplest way to introduce
confinement is to restrict the nuclear fluid to a bounded domain and
impose a particular choice of boundary conditions for the Coulombic
potential. Then, after a suitable non-dimensionalization the energy
takes the form
\begin{align} 
  \label{EOm} 
  E(u) := \int_\Omega |\nabla u| \, dx +
  \frac12 \int_\Omega \int_\Omega G(x, y) (u(x) - \bar u) (u(y) - \bar
  u) \, dx \, dy.
\end{align}
Here, $\Omega \subset \mathbb R^3$ is the spatial domain (bounded),
$u \in BV(\Omega; \{ 0, 1\})$ is the characteristic function of the
region occupied by the nuclear fluid (nuclear fluid density),
$\bar u \in (0,1)$ is the neutralizing uniform background density of
electrons, and $G$ is the Green's function of the Laplacian which, in
the case of Neumann boundary conditions for the electrostatic
potential solves
\begin{align} \label{GOm} %
  -\Delta_x G(x, y) = \delta(x - y) - {1 \over |\Omega|},
\end{align}
where $\delta(x)$ is the Dirac delta-function.  The nuclear fluid
density must also satisfy the global electroneutrality constraint:
\begin{align}
  \label{uub0}
  {1 \over |\Omega|} \int_\Omega u \, dx = \bar u.
\end{align}
In writing \eqref{EOm} we took into account that because of the
scaling properties of the Green's function one can eliminate all the
physical constants appearing in \eqref{EOm} by choosing the
appropriate energy and length scales.

It is notable that the model in \eqref{EOm}--\eqref{uub0} also appears
in a completely different physical context, namely, in the studies of
mesoscopic phases of diblock copolymer melts, where it is referred to
as the Ohta-Kawasaki model \cite{ohta86,ren00,choksi03}. This is, of
course, not surprising, considering the fundamental nature of Coulomb
forces. In fact, the range of applications of the energy in
\eqref{EOm} goes far beyond the systems mentioned above (for an
overview, see \cite{m:pre02} and references therein). Importantly, the
model in \eqref{EOm} is a paradigm for the energy-driven pattern
forming systems in which spatial patterns (global or local energy
minimizers) form as a result of the competition of short-range
attractive and long-range repulsive forces. This is why this model and
its generalizations attracted considerable attention of mathematicians
in recent years (see, e.g., \cite{choksi01, alberti09, m:cmp10,
  choksi10, choksi11, km:cpam12,km:cpam13, gms:arma13,gms:arma14,
  cicalese13,acerbi13, julin14, julin13, bonacini14, lu14, figalli15,
  mz:agag15, choksi12,frank15}, this list is certainly not
exhaustive). In particular, the volume-constrained global minimization
problem for \eqref{EOm} in the whole space with no neutralizing
background, which we will also refer to as the ``self-energy
problem'', has been investigated in
\cite{choksi10,km:cpam13,lu14,frank15}.

A question of particular physical interest is how the ground states of
the energy in \eqref{EOm} behave as the domain size tends to
infinity. In \cite{alberti09}, Alberti, Choksi and Otto showed that in
this so-called ``macroscopic'' limit the energy becomes distributed
uniformly throughout the domain. Another asymptotic regime,
corresponding to the onset of non-trivial minimizers in the
two-dimensional screened version of \eqref{EOm} was studied in
\cite{m:cmp10}, where it was shown that at appropriately low densities
every non-trivial minimizer is given by the characteristic function of
a collection of nearly perfect, identical, well separated small disks
(droplets) uniformly distributed throughout the domain (see also
\cite{gms:arma13} for a related study of almost minimizers). Further
results about the fine properties of the minimizers were obtained via
two-scale $\Gamma$-expansion in \cite{gms:arma14}, using the approach
developed for the studies of magnetic Ginzburg-Landau vortices
\cite{sandier12} (more recently, the latter was also applied to
three-dimensional Coulomb gases \cite{rougerie13}). In particular, the
method of \cite{gms:arma14} allows, in principle, to determine the
asymptotic spatial arrangement of the droplets of the low density
phase via the solution of a minimization problem involving point
charges in the plane. It is widely believed that the solution of this
problem should be given by a hexagonal lattice, which in the context
of type-II superconductors is called the ``Abrikosov lattice''
\cite{tinkham}. Proving this result rigorously is a formidable task,
and to date such a result has been obtained only within a much reduced
class of Bravais lattices \cite{chen07arma,sandier12}.

It is natural to ask what happens with the low density ground state of
the energy in \eqref{EOm} as the size of the domain $\Omega$ goes to
infinity in three space dimensions. As can be seen from the above
discussion, the answer to this question bears immediate relevance to
the structure of nuclear matter under the conditions realized in the
outer crust of neutron stars. This is the question that we address in
the present paper. On physical grounds, it is expected that at low
densities the ground state of such systems is given by the
characteristic function of a union of nearly perfect small balls
(nuclei) arranged into a body-centered cubic lattice (known to
minimize the Coulomb energy of point charges among body-centered
cubic, face-centered cubic and hexagonal close-packed lattices
\cite{fuchs35,foldy71,nagai83}). The volume of each nucleus should
maximize the binding energy per nucleon, which then yields the nucleus
of an isotope of nickel.

Our results concerning the minimizers of \eqref{EOm} proceed in that
direction, but are still far from rigorously establishing such a
detailed physical picture. One major difficulty has to do with the
lack of the complete solution of the self-energy problem
\cite{km:cpam13,choksi12}. Assuming the solution of this problem,
whenever it exists, is a spherical droplet, a mathematical conjecture
formulated by Choksi and Peletier \cite{choksi11} and a universally
accepted hypothesis in nuclear physics, we indeed recover spherical
nuclei whose volume minimizes the self-energy per unit nuclear volume
(which is equivalent to maximizing the binding energy per nucleon in
the nuclear context). The question of spatial arrangement of the
nuclei is another major difficulty related to establishing periodicity
of ground states of systems of interacting particles, which goes far
beyond the scope of the present paper.  Nevertheless, knowing that the
optimal droplets are spherical should make it possible to apply the
techniques of \cite{sandier12,rougerie13} to relate the spatial
arrangement of droplets to that of the minimizers of the renormalized
Coulomb energy.

In the absence of the complete solution of the self-energy problem, we
can still establish, although in a somewhat implicit manner, the limit
behavior of the minimizers of \eqref{EOm}--\eqref{uub0} in the case
$\Omega = \TTT_\ell$, where $\TTT_\ell$ is the three-dimensional torus
with sidelength $\ell$, as $\ell \to \infty$, provided that $\bar u$
also goes simultaneously to zero with an appropriate rate (low-density
regime). We do so by establishing the $\Gamma$-limit of the energy in
\eqref{EOm}, with the notion of convergence given by weak convergence
of measures (for a closely related study, see \cite{gms:arma13}). The
limit energy is given by the sum of a constant term proportional to
the volume occupied by the minority phase (which is also referred to
as ``mass'' throughout the paper) and the Coulombic energy of the
limit measure, with the proportionality constant in the first term
given by the minimal self-energy per unit mass among all masses for
which the minimum of the self-energy is attained. Importantly, the
minimizer of the limit energy (which is strictly convex) is given
uniquely by the uniform measure. Thus, we establish that for a
minimizer of \eqref{EOm}--\eqref{uub0} the mass in the minority phase
spreads (in a coarse-grained sense) uniformly throughout the spatial
domain and that the minimal energy is proportional to the mass, with
the proportionality constant given by the minimal self-energy per unit
mass (compare to \cite{alberti09}). We also establish that almost all
the ``droplets'', i.e., the connected components of the support of a
particular minimizer, are close to the minimizers of the self-energy
with mass that minimizes the self-energy per unit mass.

Mathematically, it would be natural to try to extend our results in
two directions. The first direction is to consider exactly the same
energy as in \eqref{EOm} in higher space dimensions. Here, however, we
encounter a difficulty that it is not known that the minimizers of the
self-energy do not exist for large enough masses. Such a result is
only available in three space dimensions for the Coulombic kernel
\cite{km:cpam13,lu14}. In the absence of such a non-existence result
one may not exclude a possibility of a network-like structure in the
macroscopic limit. Another direction is to replace the Coulombic
kernel in \eqref{EOm} with the one corresponding to a more general
negative Sobolev norm. Here we would expect our results to still hold
in two space dimensions. Furthermore, the physical picture of
identical radial droplets in the limit is expected for sufficiently
long-ranged kernels, i.e., those kernels that satisfy
$G(x, y) \sim |x - y|^{-\alpha}$ for $|x - y| \ll 1$, with
$0 < \alpha \ll 1$ \cite{km:cpam12,mz:agag15}. Note that although a
similar characterization of the minimizers for long-ranged kernels
exists in higher dimensions as well \cite{bonacini14,figalli15}, these
results are still not sufficient to be used to characterize the limit
droplets, since they do not give an explicit interval of existence of
the minimizers of the self-interaction problem. Also, since the
non-existence result for the self-energy with such kernels is
available only for $\alpha < 2$ \cite{km:cpam13}, our results may not
extend to the case of $\alpha \geq 2$ in dimensions three and above.

Finally, a question of both physical and mathematical interest is what
happens with the above picture when the Coulomb potential is screened
(e.g., by the background density fluctuations). In the simplest case,
one would replace \eqref{GOm} with the following equation defining
$G$:
\begin{align}
  \label{GOmk}
  -\Delta_x G(x, y) + \kappa^2 G(x, y) = \delta(x - y),
\end{align}
where $\kappa > 0$ is the inverse screening length, and the charge
neutrality constraint from \eqref{uub0} is relaxed. Here a bifurcation
from trivial to non-trivial ground states is expected under suitable
conditions (in two dimensions, see
\cite{m:cmp10,gms:arma13,gms:arma14}). We speculate that in certain
limits this case may give rise to non-spherical droplets that minimize
the self-energy. Indeed, in the presence of an exponential cutoff at
large distances, it may no longer be advantageous to split large
droplets into smaller disconnected pieces, and the self-energy
minimizers for arbitrarily large masses may exist and resemble a
``kebab on a skewer''. In contrast to the bare Coulomb case, in the
screened case the energy of such a kebab-shaped configuration scales
linearly with mass. Note that this configuration is reminiscent of the
pearl-necklace morphology exhibited by long polyelectrolyte molecules
in poor solvents \cite{dobrynin05,forster04}.

\paragraph{Organization of the paper.} In Sec. \ref{sec-scale}, we
introduce the specific model, the scaling regime considered, the
functional setting and the heuristics. In this section, we also
discuss the self-energy problem and mention a result about attainment
of the optimal self-energy per unit mass. In Sec. \ref{sec-main}, we
first state a basic existence and regularity result for the minimizers
(Theorem \ref{thm-exist}) and give a characterization of the
minimizers of the whole space problem that also minimize the
self-energy per unit mass (Theorem \ref{thm-fstar}).  We then state
our main $\Gamma$-convergence result in Theorem \ref{thm-main}.  In
the same section, we also state the consequences of Theorem
\ref{thm-main} to the asymptotic behavior of the minimizers in
Corollary \ref{cor-conv}, as well as Theorem \ref{thm-uniform} about
the uniform distribution of energy in the minimizers and Theorem
\ref{thm-droplets} that establishes the multidroplet character of the
minimizers. Section \ref{sec-general} is devoted to generalized
minimizers of the self-energy problem, where, in particular, we obtain
existence and uniform regularity for minimizers in Theorem
\ref{thm-mingen} and Theorem \ref{prodelta}. This section also
establishes a connection to the minimizers of the whole space problem
with a truncated Coulombic kernel and ends with a characterization of
the optimal self-energy per unit mass in Theorem
\ref{prp-fstar}. Section \ref{sec-main-proof} contains the proof of
the $\Gamma$-convergence result of Theorem \ref{thm-main} and of the
equidistribution result of Theorem \ref{thm-uniform}.  Section
\ref{sec-prop} establishes uniform estimates for the problem on the
rescaled torus, where, in particular, uniform estimates for the
potential are obtained in Theorem \ref{lem-pot}.  Section
\ref{sec-drop} presents the proof of Theorem \ref{thm-droplets}.
Finally, some technical results concerning the limit measures
appearing in the $\Gamma$-limit are collected in the Appendix.

\paragraph*{Notation.} Throughout the paper $H^1$, $BV$, $L^p$, $C^k$,
$C^k_c$, $C^{k,\alpha}$, $\mathcal M$ denote the usual spaces of
Sobolev functions, functions of bounded variation, Lebesgue functions,
functions with continuous derivatives up to order $k$, compactly
supported functions with continuous derivatives up to order $k$,
functions with H\"older-continuous derivatives up to order $k$ for
$\alpha \in (0,1)$, and the space of finite signed Radon measures,
respectively. We will use the symbol $|\nabla u|$ to denote the Radon
measure associated with the distributional gradient of a function of
bounded variation. With a slight abuse of notation, we will also
identify Radon measures with the associated, possibly singular,
densities (with respect to the Lebesgue measure) on the underlying
spatial domain. For example, we will write $\nu = |\nabla u|$ and
$d \nu(x) = |\nabla u(x)| \, dx$ to imply $\nu \in \mathcal M(\Omega)$
and
$\nu(\Omega') = |\nabla u|(\Omega') = \int_{\Omega'} |\nabla u| \,
dx$,
given $u \in BV(\Omega)$ and $\Omega' \subset \Omega$. For a set
$I \subset \mathbb N$, $\# I$ denotes cardinality of $I$. The symbol
$\chi_F$ always stands for the characteristic function of the set $F$,
and $|F|$ denotes its Lebesgue measure. We also use the notation
$(u_\eps) \in \mathcal A_\eps$ to denote sequences of functions
$u_\eps \in \mathcal A_\eps$ as $\eps = \eps_n \to 0$, where
$\mathcal A_\eps$ are admissible classes.

\section{Mathematical setting and scaling} \label{sec-scale}

\paragraph{Variational problem on the unit torus.} Throughout the rest
of this paper the spatial domain $\Omega$ in \eqref{EOm} is assumed to
be a torus, which allows us to avoid dealing with boundary effects and
concentrate on the bulk properties of the energy minimizers. We define
$\TTT := \mathbb R^3 / \mathbb Z^3$ to be the flat three-dimensional
torus with unit sidelength. For $\eps > 0$, which should be treated as
a small parameter, we introduce the following energy functional:
\begin{align}\label{Eeps}
  E_\eps(u) \ := \ \eps \int_\TTT |\nabla u| \, dx + \frac12 \int_\TTT
  (u - \bar u_\eps)(-\Delta)^{-1} (u - \bar u_\eps) \, dx,
\end{align}
where the first term is understood distributionally and the second
term is understood as the double integral involving the periodic
Green's function of the Laplacian, with $u$ belonging to the
admissible class
\begin{align}
  \label{Aeps}
  \AA_\eps \ := \ \left\{ u \in BV(\TTT; \{ 0, 1\}) \ : \ \int_\TTT u \, dx
    = \bar u_\eps \right\},
\end{align}
where
\begin{align}
  \label{ubeps}
  \bar u_\eps := \lambda \, \eps^{2/3},
\end{align}
with some fixed $\lambda > 0$.  The choice of the scaling of
$\bar u_\eps$ with $\eps$ in \eqref{ubeps} will be explained
shortly. To simplify the notation, we suppress the explicit dependence
of the admissible class on $\lambda$, which is fixed throughout the
paper.

\medskip

It is natural to define for $u \in \AA_\eps$ the measure $\mu_\eps$
by
\begin{align} \label{mueps} %
  d \mu_\eps(x) := \eps^{-{2/3}} u(x) \, dx.
\end{align}
In particular, $\mu_\eps$ is a positive Radon measure and satisfies
$\mu_\eps(\TTT) = \lambda$.  Therefore, on a suitable sequence as
$\eps \to 0$ the measure $\mu_\eps$ converges weakly in the sense of
measures to a limit measure $\mu$, which is again a positive Radon
measure and satisfies $\mu(\TTT) = \lambda$.

\paragraph{Function spaces for the measure and potential.} In terms of
$\mu_\eps$ the Coulombic term in \eqref{Eeps} is given by
\begin{align} \label{culo} %
  \frac12 \int_\TTT (u - \bar u_\eps) (-\Delta)^{-1} (u - \bar u_\eps)
  \, dx = \frac{\eps^{4/3}}{2} \int_\TTT \int_\TTT G(x - y) \, d
  \mu_\eps(x) \, d \mu_\eps(y),
\end{align}
where
$G$ is the periodic Green's function of the Laplacian on $\TTT$,
i.e., the unique distributional solution of
\begin{align}
  \label{G}
  - \Delta G(x) = \delta(x) - 1, \qquad \qquad \int_{\TTT} G(x) \, dx =
  0.
\end{align}
If the kernel $G$ in \eqref{culo} were smooth, then one would be able
to pass directly to the limit in the Coulombic term and obtain the
corresponding convolution of the kernel with the limit measure. This
is not possible due to the singularity of the kernel at $\{x =
y\}$.
In fact, the double integral involving the limit measure may be
strictly less than the $\liminf$ of the sequence, and the defect of
the limit is related to a non-trivial contribution of the
self-interaction of the connected components of the set $\{ u = 1 \}$
and its perimeter to the limit energy.

\medskip

On the other hand, the singular character of the kernel provides
control on the regularity of the limit measure $\mu$. To see this, we
define the electrostatic potential $v_\eps \in H^1(\TTT)$ by
\begin{align} \label{veps} %
  v_\eps(x) :=  \int_\TTT G(x-y) \, d \mu_\eps(y),
\end{align}
which solves
\begin{align}
  \label{vepsweak}
  \int_\TTT \nabla \varphi \cdot \nabla v_\eps \, dx = \int_\TTT
  \varphi \, d \mu_\eps - \lambda \int_\TTT \varphi \, dx \qquad
  \qquad \forall \varphi \in C^\infty(\TTT).
\end{align}
By \eqref{mueps}, we can rewrite the corresponding term in the
Coulombic energy as
\begin{align}
  \label{Gmuepsmueps0}
  \int_\TTT \int_\TTT G(x - y) \, d \mu_\eps(x) \, d \mu_\eps(y) =
  \int_\TTT 
  v_\eps \, d \mu_\eps = \int_\TTT |\nabla v_\eps|^2 \, dx.
\end{align}
Hence, if the left-hand side of \eqref{Gmuepsmueps0} remains bounded
as $\eps \to 0$, and since $\int_\TTT v_\eps \, dx = 0$, the sequence
$v_\eps$ is uniformly bounded in $H^1(\TTT)$ and hence weakly
convergent in $H^1(\TTT)$ on a subsequence.

\medskip

By the above discussion, the natural space for the potential is the
space
\begin{align} \label{calH} \mathcal H \ := \ \left\{v \in
    H^1(\TTT):\int_\TTT v \, dx =0\right\} %
  &&\text{with } && %
                    \|v\|_{\mathcal H} \ := \ \left(\int_\TTT |\nabla
                    v|^2 \, dx \right)^{1/2}.
\end{align}
The space $\mathcal H$ is a Hilbert space together with the inner
product
\begin{align}
  \label{Hinner}
  \langle u, v \rangle_{\mathcal H} := \int_\TTT \nabla u \cdot \nabla
  v \, dx \qquad \qquad \forall u,v \in \mathcal H.
\end{align}
The natural class for measures $\mu_\eps$ to consider is the class of
positive Radon measures on $\TTT$ which are also in $\mathcal H'$, the
dual of $\mathcal H$. More precisely, let
$\MM^+(\TTT) \subset \MM(\TTT)$ be the set of all positive Radon
measures on $\TTT$.  We define the subset
$\MM^+(\TTT) \cap \mathcal H'$ of $\MM^+(\TTT)$ by
\begin{align} \label{def-H1*} %
  \MM^+(\TTT) \cap \mathcal H' %
  = \Big \{ \mu \in \MM^+(\TTT) \ : \ \int_\TTT \phi\, d\mu \leq C
  \NNN{\phi}{\mathcal H} \ \ \ \forall \phi \in \mathcal H \cap
  C(\TTT) \Big\},
\end{align}
for some $C>0$. This is the set of positive Radon measures which can
be understood as continuous linear functionals on $\mathcal H$.  Note
that $\mu \in \MM^+(\TTT)$ satisfies
$\mu \in \MM^+(\TTT) \cap \mathcal H'$ if and only if it has finite
Coulombic energy, i.e.
\begin{align} \label{muCoul} %
  \int_\TTT \int_\TTT G(x - y) \, d \mu(x) \, d \mu(y) < \infty,
\end{align}
with the convention that $G(0) = +\infty$. The proof of this
characterization and related facts about
$\MM^+(\TTT) \cap \mathcal H'$ are given in the Appendix.

\paragraph{The whole space problem. } We will also consider the
following related problem, formulated on $\R^3$. We consider the
energy
\begin{align} \label{Einf} %
  \widetilde E_\infty(u) \ := \ \int_\Rn |\nabla u| \, dx + {1 \over 8
  \pi} \int_\Rn \int_\Rn {u(x) u(y) \over |x - y|} \, dx \, dy.
\end{align}
The appropriate admissible class for the energy $\widetilde E_\infty$
in the present context is that of configurations with prescribed
``mass'' $m > 0$:
\begin{align} \label{Ainf} %
  \widetilde{\mathcal A}_\infty(m) := \left\{ u \in BV(\R^3; \{ 0,1\}) \ : 
    \int_\Rn u \, dx = m \right\}.
\end{align}
For a given mass $m > 0$, we define the minimal energy by
\begin{align}\label{eqe}
  e(m) \ := \ \inf_{u \in \widetilde A_\infty(m)} \widetilde E_\infty(u).
\end{align}
The set of masses for which the infimum of $\widetilde E_\infty$ in
$\widetilde A_\infty(m)$ is attained is denoted by
\begin{align} \label{I} %
  \mathcal I := \left\{ m \geq 0 \ : \ \exists\ u_m \in \widetilde
    A_\infty(m), \ \widetilde E_\infty(u_m) = e(m) \right\} ,
\end{align}
The minimization problem associated with \eqref{Einf} and \eqref{Ainf}
was recently studied by two of the authors in \cite{km:cpam13}. In
particular, by \cite[Theorem 3.3]{km:cpam13} the set $\mathcal I$ is
bounded, and by \cite[Theorems 3.1 and 3.2]{km:cpam13} the set $\II$
is non-empty and contains an interval around the origin.

\medskip

For $m \geq 0$, we also define the quantity (with the convention that
$f(0) := +\infty$)
\begin{align} \label{f} %
  f(m) := {e(m) \over m},
\end{align}
which represents the minimal energy for \eqref{Einf} and \eqref{Ainf}
per unit mass.  By \cite[Theorem 3.2]{km:cpam13} there is a universal
$\widetilde m_0 > 0$ such that $f(m)$ is obtained by evaluating
$\widetilde E_\infty$ on a ball of mass $m$ for all
$m \leq \widetilde m_0$. After a simple computation, this yields
\begin{align}  \label{fmball} %
  f(m) = 6^{2/3} \pi^{1/3} m^{-1/3} + 3^{2/3} \cdot 2^{-1/3} \cdot
  10^{-1} \cdot \pi^{-2/3} m^{2/3} \qquad \text{for all } 0 < m \leq
  \widetilde m_0.
\end{align}
Note that obviously this expression also gives an a priori upper
  bound for $f(m)$ for all $m > 0$. In addition, by \cite[Theorem
  3.4]{km:cpam13} there exist universal constants $C,c > 0$ such that
  \begin{align}
    \label{eq:fcC}
  c \leq f(m) \leq C \qquad \text{for all} \ m \geq \widetilde m_0.
  \end{align}
  It was conjectured in \cite{choksi11} that
  $\mathcal I = [0, \widetilde m_0]$ and that
  $\widetilde m_0 = m_{c1}$, where
  \begin{align}
    \label{mc1}
    m_{c1} := \frac{40 \pi}{3} \left( 2^{1/3} + 2^{-1/3} - 1 \right)
    \approx 44.134.
  \end{align}
  The quantity $m_{c1}$ is the maximum value of $m$ for which a ball
  of mass $m$ has less energy than twice the energy of a ball with
  mass $\frac12 m$.  However, such a result is not available at
  present and remains an important challenge for the considered class
  of variational problems (for several related results see
  \cite{km:cpam12,bonacini14,mz:agag15}).

Finally, we define
\begin{align} \label{fstar} %
  f^* := \inf_{m \in \mathcal I} f(m) && %
  \text{and} && %
  \mathcal I^* := \left\{ m^* \in \mathcal I \ : \ f(m^*) = f^*
                \right\}. 
\end{align}
Observe that in view of \eqref{fmball} and \eqref{eq:fcC} we have
$f^* \in (0, \infty)$. Also, as we will show in Theorem
\ref{thm-fstar}, the set $\mathcal I^*$ is non-empty, i.e., the
minimum of $f(m)$ over $\mathcal I$ is attained. In fact, the minimum
of $f(m)$ over $\mathcal I$ is also the minimum over all
$m \in (0, \infty)$ (see Theorem \ref{prp-fstar}). Note that this
result was also independently obtained by Frank and Lieb in their
recent work \cite{frank15}.  The set $\mathcal I^*$ of masses that
minimize the energy $\widetilde E_\infty$ per unit mass and the
associated minimizers (which in general may not be unique) will play a
key role in the analysis of the limit behavior of the minimizers of
$E_\eps$. Note that if $f(m)$ were given by \eqref{fmball} and
$\mathcal I = [0, m_{c1}]$, then we would have explicitly
$\mathcal I^* = \{10 \pi\}$ and
$f^* = 3^{5/3} \cdot 2^{-2/3} \cdot 5^{-1/3} \approx 2.29893$. On the
other hand, in view of the statement following \eqref{fmball}, this
value provides an a priori upper bound on the optimal energy density.

\paragraph{Macroscopic limit \& heuristics.} The limit $\eps \to 0$
with $\lambda > 0$ fixed is equivalent to the limit of the energy in
\eqref{EOm} with $\Omega = \TTT_\ell$, where
$\TTT_\ell := \mathbb R^3 / (\ell \mathbb Z)^3$ is the torus with
sidelength $\ell > 0$, as $\ell \to \infty$. Indeed, introducing the
notation
\begin{align} \label{El} %
  \widetilde E_\ell(\tilde u) \ := \ \int_{\torol} |\nabla \tilde u|
  \, dx + \frac12 \int_\torol (\tilde u - \bar{\tilde{u}}_\ell)
  (-\Delta)^{-1} (\tilde u - \bar{\tilde{u}}_\ell) \, dx,
\end{align}
for the energy in \eqref{EOm} with $\Omega = \TTT_\ell$, and taking
$\bar{\tilde{u}}_\ell = \lambda \ell^{-2}$ and
$\t u \in \widetilde{\mathcal A}_\ell$, where
\begin{align}\label{Al}
  \widetilde{\mathcal A}_\ell \ := \ \left\{ \tilde u \in BV(\torol; \{ 0,
  1\}) \ : \ \int_\torol \tilde u \, dx = \lambda \ell
  \right\},
\end{align}
it is easy to see that $u(x) := \tilde u(\ell x)$ belongs to
$\mathcal A_\eps$ with $\bar u_\eps = \lambda \eps^{2/3}$ for
$\eps = \ell^{-3}$, and we have
$E_\eps(u) = \ell^{-5} \widetilde E_\ell(\tilde u)$. It then follows
that the two limits $\ell \to \infty$ and $\eps \to 0$ are
equivalent. Note that the full space energy $\widetilde E_\infty$ is
the formal limit of \eqref{El} for $\ell \to \infty$.

\medskip

The choice of the scaling of $\bar u_\eps$ with $\eps$ is determined
by the balance of far-field and near-field contributions of the
Coulomb energy. Heuristically, one would expect the minimizers of the
energy in \eqref{Eeps} to be given by the characteristic function of a
set that consists of ``droplets'' of size of order $R \ll 1$ separated
by distance of order $d$ satisfying $R \ll d \ll 1$ (for evidence
based on recent molecular dynamics simulations, see also
\cite{schneider13}). Assuming that the volume of each droplet scales
as $R^3$ (think, for example, of all the droplets as non-overlapping
balls of equal radius and with centers forming a periodic lattice),
from \eqref{El} we find for the surface energy, self-energy and
interaction energy, respectively, for a single droplet:
\begin{align}
  \label{eq:EsEsEi}
  E_\mathrm{surf} \sim \eps R^2, \qquad E_\mathrm{self} \sim R^5,
  \qquad E_\mathrm{int} \sim {R^6 \over d^3}.
\end{align}
Equating these three quantities and recalling that
$R^3 / d^3 \sim \bar u_\eps$, we obtain
\begin{align}
  \label{Rdscale}
  R \sim \eps^{1/3}, \qquad d \sim \eps^{1/9}, \qquad \bar u_\eps \sim
  \eps^{2/3},
\end{align}
which leads to \eqref{ubeps}. Note that, in some sense, this is the
most interesting low volume fraction regime that leads to infinitely
many droplets in the limit as $\eps \to 0$, since both the self-energy
of each droplet and the interaction energy between different droplets
contribute comparably to the energy. For other scalings one would
expect only one of these two terms to contribute in the limit, which
would, however, result in loss of control on either the perimeter term
or the Coulomb term as $\eps \to 0$ and, as a consequence, a possible
change in behavior. Let us note that a different scaling regime, in
which $\bar u_\eps = O(\eps)$, leads instead to finitely many droplets
that concentrate on points as $\eps \to 0$ \cite{choksi10}, while for
$\bar u_\eps = O(1)$ one expects phases of reduced dimensionality,
such as rods and slabs (see Fig. \ref{f:pasta}).

\section{Statement of the main results} \label{sec-main} %

We now turn to stating the main results of this paper concerning the
asymptotic behavior of the minimizers or the low energy configuration
of the energy in \eqref{Eeps} within the admissible class in
\eqref{Aeps}. Existence of these minimizers is guaranteed by the
following theorem.

\begin{theorem}[Minimizers: existence and
  regularity] \label{thm-exist} %
  For every $\lambda > 0$ and every $0 < \eps < \lambda^{-3/2}$ there
  exists a minimizer $u_\eps \in \mathcal A_\eps$ of $E_\eps$ given by
  \eqref{Eeps} with $\bar u_\eps$ given by \eqref{ubeps}. Furthermore,
  after a possible modification of $u_\eps$ on a set of zero Lebesgue
  measure the support of $u_\eps$ has boundary of class $C^\infty$.
\end{theorem}

\begin{proof}
  The proof of Theorem \ref{thm-exist} is fairly standard. We present
  a few details below for the sake of completeness.

  By the direct method of the calculus of variations, minimizers of
  the considered problem exist for all $\eps > 0$ as soon as the
  admissible class $\mathcal A_\eps$ is non-empty, in view of the fact
  that the first term is coercive and lower semicontinuous in
  $BV(\TTT)$, and that the second term is continuous with respect to
  the $L^1(\TTT)$ convergence of characteristic functions. The
  admissible class is non-empty if and only if
  $\eps < \lambda^{-3/2}$.

  H\"older regularity of minimizers was proved in \cite[Proposition
  2.1]{sternberg11}, where it was shown that the essential support of
  minimizers has boundary of class $C^{3,\alpha}$. Smoothness of the
  boundary was established in \cite[Proposition 2.2]{julin13} (see
  also the proof of Lemma \ref{lemreg} below for a brief outline of
  the argument in a closely related context).
\end{proof}

In view of the regularity statement above, throughout the rest of the
paper we always choose the regular representative of a minimizer.

\medskip

We proceed by giving a characterization of the quantity $f^*$ defined
in \eqref{fstar} as the minimal self-energy of a single droplet per
unit mass, i.e., as the minimum of $f(m)$ over $\mathcal I$.

\begin{theorem}[Self-energy: attainment of optimal energy per unit
  mass]\label{thm-fstar}
Let $f^*$ be defined as in \eqref{fstar}. Then
  there exists $m^* \in \mathcal I$ such that $f^* = f(m^*)$.
\end{theorem}

\noindent With the result in Theorem \ref{thm-fstar}, we are now in
the position to state our main result on the $\Gamma$-limit of the
energy in \eqref{Eeps}, which can be viewed as a generalization of
\cite[Theorem 1]{gms:arma13}.

\begin{theorem}[$\Gamma$-convergence] \label{thm-main} %
  For a given $\lambda > 0$, let $E_\eps$ be defined by \eqref{Eeps}
  with $\bar u_\eps$ given by \eqref{ubeps}. Then as $\eps \to 0$ we
  have $\eps^{-{4/3}} E_\eps \stackrel{\Gamma}{\to} E_0$, where
  \begin{align} \label{E0} %
    E_0(\mu) := \lambda f^* + \frac12 \int_\TTT \int_\TTT G(x-y) \, d
    \mu(x) \, d \mu(y),
  \end{align}
  and
  $\mu \in \MM^+(\TTT) \cap \mathcal H'$ satisfies
  $\mu (\TTT) = \lambda$. More precisely,
  \begin{itemize}
  \item[i)] (Compactness and $\Gamma$-liminf inequality) Let
    $(u_\eps) \in \mathcal A_\eps$ be such that
    \begin{align}
      \label{Epsbounded}
      \limsup_{\eps \to 0} \eps^{-4/3} E_\eps(u_\eps) < \infty,
    \end{align}
    and let $\mu_\eps$ and $v_\eps$ be defined in \eqref{mueps} and
    \eqref{veps}, respectively.  Then, upon extraction of a
    subsequence, we have 
    \begin{align} \label{muepsmu} %
      \mu_\eps \rightharpoonup \mu \textrm{ in } \MM(\TTT), \qquad
      v_\eps \rightharpoonup v \textrm{ in } \mathcal H,
    \end{align}
    as $\eps \to 0$, for some $\mu \in \MM^+(\TTT) \cap \mathcal H'$
    with $\mu(\TTT) = \lambda$, the function $v$ has a representative
    in $L^1(\TTT, d \mu)$ given by
    \begin{align} \label{v} %
      v(x) = \int_\TTT G(x - y) \, d \mu(y),
    \end{align}
    and 
    \begin{align}
      \liminf_{\eps \to 0} \eps^{-4/3} E_\eps(u_\eps) \geq
      E_0(\mu).      
    \end{align}
  \item[ii)] ($\Gamma$-limsup inequality) For any measure $\mu \in
    \MM^+(\TTT) \cap \mathcal H'$ with $\mu(\TTT) = \lambda$ there
    exists a sequence $(u_\eps) \in \mathcal A_\eps$ such that
    \eqref{muepsmu} and \eqref{v} hold as $\eps \to 0$ for $\mu_\eps$
    and $v_\eps$ defined in \eqref{mueps} and \eqref{veps},
    and 
    \begin{align}
      \limsup_{\eps \to 0} \eps^{-4/3} E_\eps(u_\eps) \leq
      E_0(\mu).      
    \end{align}
  \end{itemize}
\end{theorem}

Note that the weak convergence of measures was recently identified in
\cite{gms:arma13} (see also \cite{m:cmp10}) as a suitable notion of
convergence for the studies of the $\Gamma$-limit of the
two-dimensional version of the energy in \eqref{Eeps}.

Observe also that the limit energy $E_0$ is a strictly convex
functional of the limit measure and, hence, attains a unique global
minimum. By direct inspection, $E_0$ is minimized by $\mu = \mu_0$,
where $d \mu_0 := \lambda dx$.  Thus, the quantity $f^*$ plays the
role of the optimal energy density in the limit $\eps \to 0$.

The remaining results are concerned with sequences of minimizers.  We
will hence assume that the functions $(u_\eps) \in \AA_\eps$ are
minimizers of the functional $E_\eps$. In this case, we can give a
more precise characterization for the asymptotic behavior of the
sequence.  We first note the following immediate consequence of
Theorem \ref{thm-main} for the convergence of sequences of minimizers.

\begin{corollary}[Minimizers: uniform distribution of
  mass] \label{cor-conv} %
  For $\lambda > 0$, let $(u_\eps) \in \mathcal A_\eps$ be minimizers
  of $E_\eps$, and let $\mu_\eps$ and $v_\eps$ be defined in
  \eqref{mueps} and \eqref{veps}, respectively. Then 
  \begin{align}
    \label{muepsveps}
    \mu_\eps \rightharpoonup \mu_0 \textrm{ in }
    \mathcal M(\TTT), \qquad v_\eps \rightharpoonup 0 \textrm{ in
    } \mathcal H, 
  \end{align}
  where $d \mu_0 = \lambda dx$, and
  \begin{align}
    \label{Eepsmin}
    \eps^{-4/3} E_\eps(u_\eps) \to \lambda f^*,
  \end{align}
  where $f^*$ is as in \eqref{fstar}, as $\eps \to 0$.
\end{corollary}

The formula in \eqref{Eepsmin} suggests that in the limit the energy
of the minimizers is dominated by the self-energy, which is captured
by the minimization problem associated with the energy
$\widetilde E_\infty$ defined in \eqref{Einf}. Therefore, it would be
natural to expect that asymptotically every connected component of a
minimizer is close to a minimizer of $\widetilde E_\infty$ under the
mass constraint associated with that connected component. Note that in
a closely related problem in two space dimensions such a result was
established in \cite{m:cmp10} for minimizers, and in
\cite{gms:arma13,gms:arma14} for almost minimizers. The situation is,
however, unique in two space dimensions, because the non-local term in
some sense decouples from the perimeter term. Hence, the minimizers
behave as almost minimizers of the perimeter and, therefore, are close
to balls. In three dimensions, however, the perimeter and the
non-local term of the self-energy $\widetilde E_\infty$ are fully
coupled, and, therefore, rigidity estimates for the perimeter
functional alone \cite{fusco08} may not be sufficient to conclude
about the ``shape'' of the minimizers. Nevertheless, we are able to
prove a result about the uniform distribution of the energy density of
the minimizers as $\eps \to 0$ in the spirit of that of
\cite{alberti09}. For a minimizer $u_\eps$, the energy density is
  associated with the Radon measure $\nu_\eps$ defined by
\begin{align}
  \label{nuepsden}
  d \nu_\eps := \eps^{-4/3} \Big( \eps |\nabla u_\eps| + \tfrac12
  \eps^{2/3} u_\eps v_\eps \Big) dx,
\end{align}
where $v_\eps$ is given by \eqref{veps} and \eqref{mueps}.
Furthermore, we are able to identify the leading order constant in the
asymptotic behavior of the energy density. 

\begin{theorem}[Minimizers: uniform distribution of
  energy] \label{thm-uniform} %
  For $\lambda > 0$, let $(u_\eps) \in \mathcal A_\eps$ be minimizers
  of $E_\eps$ and let $\nu_\eps$ be defined in \eqref{nuepsden}. Then
  \begin{align}
    \label{uepsmuepsveps}
    \nu_\eps \rightharpoonup \nu_0
    \textrm{ in } \MM(\TTT) \qquad \text{ as $\eps \to 0$,}
    \end{align}
    where $d \nu_0 = \lambda f^* dx$ and $f^*$ is as in \eqref{fstar}.
\end{theorem}

Finally, we characterize the connected components of the support of
the minimizers of $E_\eps$ and show that almost all of them approach,
on a suitable sequence as $\eps \to 0$ and after a suitable rescaling
and translation, a minimizer of $\widetilde E_\infty$ with mass in the
set $\mathcal I^*$.

\begin{theorem}[Minimizers: droplet structure] \label{thm-droplets} %
  For $\lambda > 0$, let $(u_\eps) \in \mathcal A_\eps$ be regular
  representatives of minimizers of $E_\eps$, let $N_\eps$ be the
  number of the connected components of the support of $u_\eps$, let
  $u_{\eps,k} \in BV(\R^3; \{0, 1\})$ be the characteristic function
  of the $k$-th connected component of the support of the periodic
  extension of $u_\eps$ to the whole of $\R^3$ modulo translations in
  $\mathbb Z^3$, and let $x_{\eps,k} \in \mathrm{supp}(u_{\eps,k})$.
  Then there exists $\eps_0 > 0$ such that the following properties
  hold:
  \begin{enumerate}
  \item[i)] There exist universal constants $C, c > 0$ such that for
    all $\eps \leq \eps_0$ we have
  \begin{align}
    \label{ccDueps}
    \| v_\eps \|_{L^\infty(\TTT)} \leq C \qquad \mathrm{and} \qquad
    \int_{\R^3} u_{\eps,k} \, dx \geq c \eps,
  \end{align}
  where $v_\eps$ is given by \eqref{veps}.
\item[ii)] There exist universal constants $C, c > 0$ such that for
  all $\eps \leq \eps_0$ we have
  \begin{align}
    \label{cCNeps}
    \mathrm{supp} (u_{\eps,k}) \subseteq B_{C \eps^{1/3}}(x_{\eps,k})
    \qquad \mathrm{and} \qquad c  \lambda   
    \eps^{-1/3} \leq N_\eps \leq C 
    \lambda \eps^{-1/3} . 
  \end{align}
\item[iii)] There exists $\widetilde N_\eps \leq N_\eps$ with
  $\widetilde N_\eps / N_\eps \to 1$ as $\eps \to 0$ and a subsequence
  $\eps_n \to 0$ such that for every $k_n \leq \widetilde N_{\eps_n}$
  the following holds: After possibly relabeling the connected
  components, we have
  \begin{align}
    \tilde u_n \to \tilde u &&\text{in $L^1(\R^3)$}, 
  \end{align}
  where
  $\tilde u_n (x) := u_{\eps_n,k_n}(\eps_n^{1/3} (x +
  x_{\eps_n,k_n}))$,
  and $\tilde u$ is a minimizer of $\widetilde E_\infty$ over
  $\widetilde{\mathcal A}_\infty(m^*)$ for some
  $m^* \in \mathcal I^*$, where $\mathcal I^*$ is defined in
  \eqref{fstar}.
  \end{enumerate}
\end{theorem}

The significance of this theorem lies in the fact that it shows that
all the connected components of the support of a minimizer for
sufficiently small $\eps$ look like a collection of {\em droplets} of
size of order $\eps^{1/3}$ separated by distances of order
$\eps^{1/9}$ on average. In particular, the conclusion of the theorem
excludes configurations that span the entire length of the torus, such
as the ``spaghetti'' or ``lasagna'' phases of nuclear pasta (see
Fig. \ref{f:pasta}). Thus, the ground state for small enough
$\eps > 0$ is a multi-droplet pattern (a ``meatball''
phase). Furthermore, after a rescaling most of these droplets converge
to minimizers of the non-local isoperimetric problem associated with
$\widetilde E_\infty$ that minimize the self-energy per unit mass.

\section{The problem in the whole space} \label{sec-general}

In this section, we derive some results about the single droplet
problem from \eqref{Einf}--\eqref{Ainf} and the rescaled problem from
\eqref{El}--\eqref{Al}.

\subsection{The truncated energy $\widetilde E_\infty^R$}

For reasons that will become apparent shortly, it is helpful to
consider the energies where the range of the nonlocal interaction is
truncated at certain length scale $R$. We choose a cut-off function
$\eta \in C^\infty(\R)$ with $\eta'(t) \leq 0$ for all $t\in\R$,
$\eta(t)=1$ for all $t \leq 1$ and $\eta(t)=0$ for all $t\geq 2$.  In
the following, the choice of $\eta$ is fixed once and for all, and the
dependence of constants on this choice is suppressed to avoid clutter
in the presentation. For $R > 0$, we then define
$\eta_R \in C^\infty(\mathbb R^3)$ by $\eta_R(x) := \eta(|x|/R)$. For
$u \in \widetilde{\mathcal A}_\infty(m)$, we consider the truncated
energy
\begin{align} \label{Einfk} %
  \widetilde E_\infty^R(u) := \int_{\R^3} |\nabla u| \, dx
  +\int_{\R^3} \int_{\R^3} \frac{\eta_R ( x-y ) u(x)u(y)}{8 \pi
  |x-y|}\, dx \, dy.
\end{align}
This functional will be useful in the analysis of the variational
problems associated with $\widetilde E_{\infty}$ and $E_\eps$. We
recall that by the results of \cite{rigot00}, for each $R > 0$ and
each $m > 0$ there exists a minimizer of $\widetilde E_\infty^R$ in
$\widetilde{\mathcal A}_\infty(m)$.  Furthermore, after a possible
redefinition on a set of Lebesgue measure zero, its support has
boundary of class $C^{1,1/2}$ and consists of finitely many connected
components. Below we always deal with the representatives of
minimizers that are regular.

\medskip

The following uniform density bound for minimizers of the energy is an
adaption of \cite[Lemma 4.3]{km:cpam13} for the truncated energy
$\widetilde E_\infty^R$ and generalizes the corresponding bound for
minimizers of $\widetilde E_{\infty}$.

\begin{lemma}[Density bound] \label{lem-uER} %
  There exists a universal constant $c > 0$ such that for every
  minimizer $u \in \widetilde{\mathcal A}_\infty(m)$ of
  $\widetilde E_\infty^R$ for some $R,m > 0$ and any
  $x_0 \in \overline F$ we have
  \begin{align}
    \int_{B_r(x_0)} u \, dx\geq cr^3 \qquad\text{for all }r\le
    \min(1,m^{1/3}).    
  \end{align}
\end{lemma}
\begin{proof}
  The claim follows by an adaption of the proof of \cite[Lemma
  4.3]{km:cpam13} to our truncated energy
  $\widetilde E_\infty^R$. Indeed, it is enough to show that the
  statement of \cite[Lemma 4.2]{km:cpam13} holds with
  $\widetilde E_{\infty}$ replaced by $\widetilde E_\infty^R$.
  The proof of this statement needs to be modified, since the kernel
  in the definition of $\widetilde E_\infty^R$ is not
  scale-invariant. We sketch the necessary changes, using the same
  notation as in \cite{km:cpam13}.

  \medskip
  
  The construction of the sets $\widetilde F$ and $\widehat F$
  proceeds as in the proof of \cite[Lemma 4.3]{km:cpam13}. The upper
  bound \cite[Eq.~(4.6)]{km:cpam13} still holds since
  $\widetilde E_\infty^R(u) \leq \widetilde E_{\infty}(u)$.  Related
  to the cut-off function in the definition of
  $\widetilde E_\infty^R$, we get an additional term in the right-hand
  side of the first line of \cite[Eq.~(4.6)]{km:cpam13}, which is of
  the form
  \begin{align} \label{extratermER} %
    \int_{\ell F_1} \int_{\ell F_1} {\eta_R(x - y) \over |x -
      y|^\alpha} \, dx \, dy - \ell^{2n-\alpha} \int_{F_1} \int_{F_1}
    {\eta_R(x - y) \over
      |x - y|^\alpha} \, dx \, dy \hspace{2cm} \notag \\
    = \ell^{2n-\alpha} \int_{F_1} \int_{F_1} {\eta_{R / \ell}(x - y) -
      \eta_R(x - y) \over |x - y|^\alpha} \, dx \, dy < 0,
  \end{align}
  since $\ell > 1$ and since the function $\eta$ is monotonically
  decreasing (note that $\alp = 1$ in our case). Since this term has a
  negative sign, \cite[Eq.~(4.6)]{km:cpam13} still holds. The rest of
  the argument then carries through unchanged.
\end{proof}
The following lemma establishes a uniform diameter bound for the
minimizers of $\widetilde E_\infty^R$. The idea of the proof is
similar to the one in \cite[Lemma 5]{lu14}.
\begin{lemma}[Diameter bound] \label{lem-ln} %
  There exist universal constants $R_0 > 0$ and $D_0 > 0$ such that
  for any $R \geq R_0$, any $m > 0$ and for any minimizer
  $u \in \widetilde{\mathcal A}_\infty(m)$ of
  $\widetilde E_\infty^R$, the diameter of each connected
  component $F_0$ of $\mathrm{supp}(u)$ is bounded above by $D_0$.
\end{lemma}
\begin{proof}
  Let $F_0$ be a connected component of the support of $u$ with
  $m_0 := |F_0|$. Since $u$ is a minimizer, $\chi_{F_0}$ is also a
  minimizer of $\widetilde E_\infty^R$ over
  $\mathcal A_\infty(m_0)$. Indeed, if not, replacing $u$ with
  $u - \chi_{F_0} + \chi_{\widetilde F_0}$, where
  $\chi_{\widetilde F_0}$ is a minimizer of
  $\widetilde E_\infty^R$ over $\mathcal A_\infty(m_0)$
  translated sufficiently far from the support of $u$ would lower the
  energy, contradicting the minimizing property of $u$.

  \medskip

  We may assume without loss of generality that $R \geq 2$ and
  $\diam F_0 \geq 2$. Then there is $N \in \N$ such that
  $2N \leq \diam F_0 < 2(N+1)$.  In particular there exist
  $x_0, \ldots, x_N \in \overline F_0$ such that $|x_k - x_0| = 2 k$
  for every $1 \leq k \leq N$ and, therefore, the balls $B_1(x_k)$ are
  mutually disjoint. If $m_0 \leq 1$, then by Lemma \ref{lem-uER} we
  have $|F_0 \cap B_r(x_k)| \geq c m_0$ for $r = m_0^{1/3} \leq 1$ and
  some universal $c > 0$. Therefore,
  \begin{align}
    \label{m0l1}
    m_0 \geq \sum_{k=1}^N |F_0 \cap B_r(x_k)| \geq c m_0 N,
  \end{align}
  implying that $N \leq N_0$ for some universal $N_0 \geq 1$ and,
  hence, $\diam F_0 \leq 2(N_0 + 1)$. If, on the other hand,
  $m_0 > 1$, then by Lemma \ref{lem-uER} we have
  $|F_0 \cap B_1(x_k)| \geq c$ for some universal $c > 0$.  By
  monotonicity of the kernel in $|x - y|$, we get
  \begin{align} \notag %
    \int_{F_0 \cap B_1(x_0)} \int_{F_0 \backslash B_1(x_0)} {\eta_R(x
    - y) \over |x - y|} \, dx \, dy %
    \geq c^2 \sum_{k=1}^N {\eta_R(2 k + 2) \over 2 k + 2} %
    \geq C \min \{ \log N, \log R \},
  \end{align}
  for some universal $C > 0$. Hence, if $R$ and $N$ are sufficiently
  large, then it is energetically preferable to move the charge in
  $B_1(x_0)$ sufficiently far from the remaining charge. More
  precisely, consider
  $\tilde u = u - \chi_{F_0 \cap B_1(x_0)} + \chi_{F_0 \cap B_1(x_0)}
  (\cdot + b)$,
  for some $b \in \mathbb R^3$ with $|b|$ sufficiently large. Then
  $\tilde u \in \mathcal A_\infty(m_0)$ and
  \begin{align}
    \label{EF0EE}
    \widetilde E_\infty^R(\tilde u) \leq \widetilde E_\infty^R(u)
    + 4 \pi - \tfrac12 C \min \{ \log N, \log R \} < 0,
  \end{align}
  for all $R \geq R_0$ and $N > N_0$ for some universal constants
  $R_0 \geq 2$ and $N_0 \geq 1$.  Therefore, minimality of $u$ implies
  that $N \leq N_0$ whenever $R \geq R_0$ and hence
  $\diam F_0 \leq 2 (N_0 + 1)$.
\end{proof}

\subsection{Generalized minimizers of $\widetilde E_\infty$}

We begin our analysis of $\widetilde E_\infty$ by introducing the
notion of generalized minimizers of the non-local isoperimetric
problem.

\begin{definition}[Generalized minimizers] \label{defgenmin} %
  Given $m>0$, we call a {\em generalized minimizer} of
  $\widetilde E_\infty$ in $\widetilde A_\infty(m)$ a collection of
  functions $(u_1,\ldots,u_N)$ for some $N \in \mathbb N$ such that
  $u_i$ is a minimizer of $\widetilde E_\infty$ over
  $\widetilde{\mathcal A}_\infty(m_i)$ with $m_i= \int_\TTT u_i \, dx$
  for all $i\in\{1,\ldots,N\}$, and
\begin{align}
  \label{msumem}
  m=\sum_{i=1}^N m_i \qquad \qquad
  \text{and} \qquad \qquad
  e(m)=\sum_{i=1}^N e(m_i).
\end{align}
\end{definition}

\noindent Clearly, every minimizer of $\widetilde E_\infty$ in
$\widetilde{\mathcal A}_\infty(m)$ is also a generalized minimizer
(with $N = 1$). As was shown in \cite{km:cpam13}, however, minimizers
of $\widetilde E_\infty$ in $\widetilde{\mathcal A}_\infty(m)$ may not
exist for a given $m > 0$ because of the possibility of splitting
their support into several connected components and moving those
components far apart. As we will show below, this possible loss of
compactness of minimizing sequences can be compensated by considering
characteristic functions of sets whose connected components are
``infinitely far apart'' and among which the minimum of the energy is
attained (by a generalized minimizer with some $N > 1$).  We also
remark that, if $(u_1, \ldots, u_N)$ is a generalized minimizer, then,
as can be readily seen from the definition, any sub-collection of
$u_i$'s is also a generalized minimizer with the mass equal to the sum
of the masses of its components.

\medskip

We now proceed to demonstrating existence of generalized minimizers of
$\widetilde E_\infty$ for all $m > 0$. We start by stating the basic
regularity properties of the minimizers of $\widetilde E_\infty$ and
the associated Euler-Lagrange equation.
\begin{lemma}[Regularity and Euler-Lagrange equation]\label{lemreg} %
  For $m > 0$, let $u$ be a minimizer of $\widetilde E_\infty$ in
  $\widetilde{\mathcal A}_\infty(m)$, and let
  $F = \mathrm{supp} \, (u)$.  Then, up to a set of Lebesgue measure
  zero, the set $F$ is a bounded connected set with boundary of class
  $C^\infty$, and we have
\begin{align}\label{eqEL}
  2 \kappa(x) +  v_F(x) = \lambda_F \qquad \text{for }x\in\partial F, 
\end{align}
where $\lambda_F\in\R$ is a Lagrange multiplier, $\kappa(x)$ is the
mean curvature of $\partial F$ at $x$ (positive if $F$ is convex), and
\begin{align} \label{vF} %
  v_F(x) := \frac{1}{4\pi}\int_{F}\frac{dy}{|x-y|}.
\end{align}
Moreover, if $m \in [m_0, m_1]$ for some $0 < m_0 < m_1$, then
$v_F \in C^{1,\alpha}(\mathbb R^3)$ and $\partial F$ is of class
$C^{3,\alpha}$, for all $\alpha\in (0,1)$, uniformly in $m$.
\end{lemma}

\begin{proof} {From} \cite[Proposition 2.1 and Lemma 4.1]{km:cpam13}
  it follows that, up to a set of Lebesgue measure zero, the set $F$
  is bounded and connected, and $\partial F$ is of class $C^{1,1/2}$.
  Since the function $v_F$ is the unique solution of the elliptic
  problem $-\Delta v = \chi_F$ with $v(x) \to 0$ for $|x|\to \infty$,
  by \cite[Lemma 4.4]{km:cpam13} and elliptic regularity theory
  \cite{gilbarg} it follows that $v_F \in C^{1,\alpha}(\mathbb R^3)$
  for all $\alpha\in (0,1)$, uniformly in $m \in [m_0, m_1]$.  The
  Euler-Lagrange equation \eqref{eqEL} can be obtained as in
  \cite[Theorem 2.3]{choksi07} (see also
  \cite{m:pre02,sternberg11}). Further regularity of $\partial F$
  follows from \cite[Proposition 2.1]{sternberg11} and
  \cite[Proposition 2.2]{julin13}.
\end{proof}

Similarly, if $(u_1,\ldots,u_N)$ is a generalized minimizer of
$\widetilde E_\infty$ and $F_i := \mathrm{supp} \, (u_i)$ for
$i\in\{1,\ldots,N\}$, the following Euler-Lagrange equation holds:
\begin{align}\label{eqELgen}
  2 \kappa_i(x) + \frac{1}{4\pi}\int_{F_i}\frac{dy}{|x-y|} =
  \lambda \qquad
  x\in\partial F_i, 
\end{align}
where $\kappa_i$ is the mean curvature of $\partial F_i$ (positive if
$F_i$ is convex) and $\lambda\in\R$ is a Lagrange multiplier
independent of $i$.

\medskip

In contrast to minimizers, generalized minimizers of
$\widetilde E_\infty$ in $\widetilde{\mathcal A}_\infty(m)$ exist for
all $m>0$:
\begin{theorem}[Existence of generalized
  minimizers] \label{thm-mingen} %
  For any $m \in (0,\infty)$ there exists a generalized minimizer
  $(u_1, \ldots, u_N)$ of $\widetilde E_\infty$ in
  $\widetilde{\mathcal A}_\infty(m)$.  Moreover, after a possible
  modification on a set of Lebesgue measure zero, the support of each
  component $u_i$ is bounded, connected and has boundary of class
  $C^\infty$.
\end{theorem}
\begin{proof}
  We may assume that $m\geq \widetilde m_0$, where
  $\widetilde m_0 > 0$ was defined in Sec. \ref{sec-scale}, since
  otherwise the minimum of $\widetilde E_\infty$ is attained by a ball
  \cite[Theorem 3.2]{km:cpam13} and the statement of the theorem holds
  true. In \cite[Theorems 5.1.1 and 5.1.5]{rigot00}, it is proved that
  the functional $\widetilde E_\infty^R$ admits a minimizer
  $u = \chi_{F_R}\in \widetilde{\mathcal A}_\infty(m)$,
  $F_R \subset \R^3$, for any $R > 0$, and after a possible
  redefinition on a set of Lebesgue measure zero, the set $F_R$ is
  regular, in the sense that it is a union of finitely many connected
  components whose boundaries are of class $C^{1,1/2}$.  Let
  $F_1, \ldots, F_N$ $\subset \R^3$ be the connected components of
  $F_R$. By Lemma \ref{lem-uER}, we have $N \leq N_0$,
  $|F_k| \geq \delta_0$ and $\diam F_k \leq D_0$ for all
  $1 \leq k \leq N$ and for some $N_0 \geq 1$ and some constants
  $D_0,\delta_0 > 0$ depending only on $m$. Furthermore, we have
  \begin{align} \label{dGiGj} %
    {\rm dist}(F_i,F_j)\geq 2 R\qquad \text{for }i\ne j,
  \end{align}
  since otherwise it would be energetically preferable to increase the
  distance between the components.  In particular, if $R \geq D_0$ the
  family of sets $F_1, \ldots, F_N$ $\subset \R^3$ generates a
  generalized minimizer $(u_1, \ldots, u_N)$ of $\widetilde E_\infty$
  by letting $u_i := \chi_{F_i}$.  Indeed, we have
  \begin{align}
    \label{eEEke}
    e(m) \geq \inf_{|F| = m} \widetilde E_\infty^R(u) =
    \sum_{i=1}^N \widetilde E_\infty^R (u_i) = 
    \sum_{i=1}^N \widetilde E_{\infty} (u_i) \geq \sum_{i=1}^N e(|F_i|)
    \geq e(m), 
  \end{align}
  and so all the inequalities in \eqref{eEEke} are in fact
  equalities.  Since $\widetilde E_{\infty} (\chi_{F_i}) \geq e(|F_i|)$
  for each $1 \leq i \leq N$, from \eqref{eEEke} we obtain that each
  set $F_i$ is a minimizer of $\widetilde E_\infty$ in
  $\widetilde A_\infty(|F_i|)$. By Lemma
  \ref{lemreg}, each set $F_i$ is bounded and connected, and $\p F_i$
  are of class $C^\infty$.
\end{proof}

The arguments in the proof of the previous theorem in fact show the
following relation between minimizers of the truncated energy
$\widetilde E_\infty^R$ and generalized minimizers of
$\widetilde E_{\infty}$.
\begin{corollary}[Generalized minimizers as minimizers of the
  truncated problem] \label{cor-ER} %
  Let $m > 0$ and $R > 0$, let
  $u \in \widetilde{\mathcal A}_\infty(m)$ be a minimizer of
  $\widetilde E_{\infty}^{R}$, and let $u = \sum_{i=1}^N u_i$, where
  $u_i$ are the characteristic functions of the connected components
  of the support of $u$. Then there exists a universal constant
  $R_1 > 0$ such that if $R \geq R_1$, then $(u_1, \ldots, u_N)$ is a
  generalized minimizer of $\widetilde E_\infty$ in
  $\widetilde{\mathcal A}_\infty(m)$.
\end{corollary}
 \begin{proof}
   We choose $R_1 = \max\{ R_0, D_0 \}$, where $R_0$ and $D_0$ are as
   in Lemma \ref{lem-ln}. Then we have
   $\widetilde E_\infty^R(\chi_{F_0}) = \widetilde
   E_{\infty}(\chi_{F_0})$
   for every connected component $F_0$ of the minimizer. With the same
   argument as the one used in the proof of Theorem \ref{thm-mingen},
   this yields the claim.
\end{proof}

We now provide some uniform estimates for generalized minimizers.

\begin{theorem}[Uniform estimates for generalized
  minimizers]\label{prodelta} %
  There exist universal constants $\delta_0 >0$ and $D_0 > 0$ such
  that, for any $m > \widetilde m_0$, where $\widetilde m_0$ is
  defined in Sec. \ref{sec-scale}, the support of each component of a
  generalized minimizer of $\widetilde E_\infty$ in
  $\widetilde{\mathcal A}_\infty(m)$ has volume bounded below by
  $\delta_0$ and diameter bounded above by $D_0$ (after possibly
  modifying the components on sets of Lebesgue measure zero).
  Moreover, there are universal constants $C, c > 0$ such that the
  number $N$ of the components satisfies
  \begin{align}  \label{mNm} %
    c m \leq N \leq C m.
  \end{align}
\end{theorem}

\begin{proof}
  Let $m \geq \widetilde m_0$ and let
  $(\chi_{F_1}, \ldots, \chi_{F_N})$ be a generalized minimizer of
  $\widetilde E_\infty$ in $\widetilde{\mathcal A}_\infty(m)$, taking
  all sets $F_i$ to be regular.  By \cite[Theorem 3.3]{km:cpam13} we
  know that there exists a universal
  $\widetilde m_2 \geq \widetilde m_0$ such that
  \begin{align}\label{estmax}
    |F_i|\leq \widetilde m_2 \qquad \text{for all }i\in\{1,\ldots,N\}\,.
  \end{align}
  Then by \cite[Lemma 4.3]{km:cpam13} and the argument of 
\cite[Lemma 4.1]{km:cpam13} we have
  \begin{align}
    \textrm{diam}(F_i) \leq D_0,
  \end{align}
  for some universal $D_0 > 0$.  On the other hand, we claim that
  taking $R \geq D_0$ we have that
  \begin{align}
    u(x) := \sum_{i=1}^N \chi_{F_i}(x + 4 i R e_1),
  \end{align}
  where $e_1$ is the unit vector in the first coordinate direction, 
  is a minimizer of $\widetilde E_\infty^R$ in
  $\widetilde A_\infty(m)$. Indeed, since the connected components of
  the support of $u$ are separated by distance $2R$, we have 
  \begin{align}
    \widetilde E_\infty^R(u) = \sum_{i=1}^N \widetilde
    E_\infty^R(\chi_{F_i}) = \sum_{i=1}^N \widetilde
    E_\infty(\chi_{F_i})  = e(m). 
  \end{align}
  At the same time, by the argument in the proof of Theorem
  \ref{thm-mingen} we have
  $\inf_{u \in \widetilde A_\infty(m)} \widetilde E_\infty^R(u)
  = e(m)$
  for all $R$ sufficiently large depending on $m$. Hence, $u$ is a
  minimizer of $\widetilde E_\infty^R$ in
  $\widetilde A_\infty(m)$ for large enough $R$. The universal lower
  bound $|F_i| \geq \delta_0$ then follows from Lemma \ref{lem-uER}
  and our assumption on $m$.

  \smallskip

  Finally, the lower bound in \eqref{mNm} is a consequence of
  \eqref{estmax}, while the upper bound follows directly from the
  lower bound on the volume of the components just obtained.
\end{proof}

\subsection{Properties of the function $e(m)$}

In this section, we discuss the properties of the functions
$e(m) = \inf_{u \in \widetilde \AA_\infty(m)} \widetilde E_\infty(u)$
and $f(m) = e(m)/m$, in particular their dependence on $m$.

\medskip

We start by showing that $e(m)$ is locally Lipschitz continuous on
$(0,\infty)$.
\begin{lemma}[Lipschitz continuity of $e$] \label{lem-elip}
  The function $e(m)$ is Lipschitz continuous on compact subsets of
  $(0,\infty)$.
\end{lemma}
\begin{proof}
  Let $m,m'\in [m_0,m_1]\subset (0,\infty)$ and let $(u_1,\ldots,u_N)$
  be a generalized minimizer of $\widetilde E_\infty$ in
  $\widetilde \AA_\infty(m)$. For $\lambda=(m'/m)^{1/3}$, we define
  the rescaled functions $u_i^\lambda$ with
  $u_i^{\lambda}(x) = u_i(\lam^{-1} x)$.  For sufficiently large
  $R > 0$, we define $u^\lam \in \widetilde \AA_\infty(m')$ by
  $u^\lam(x) := \sum_{i=1}^N u_i(\lam^{-1} x + i R e_1)$, where $e_1$
  is the unit vector in the first coordinate direction. We then have
  \begin{align} \label{Elam} %
    \widetilde E_\infty(u^\lambda) %
    = \lambda^2 \sum_{i=1}^N \int_{\R^3} |\nabla u_i| \, dx +
    \lambda^5 \sum_{i=1}^N \int_{\R^3}\int_{\R^3}\frac{u_i(x)u_i(y)}{8
    \pi |x-y|} \, dx \, dy + g(R),
  \end{align}
  where the term $g(R)$ refers to the interaction energy between
  different components $u_i^\lam$, $u_j^\lam$, $i \neq j$, of
  $u^\lam$. Clearly, we have $g(R) \to 0$ for $R \to \infty$.  It
  follows that
  \begin{multline}  
    \widetilde E_\infty(u^\lambda) - e(m) \leq
    \left|\lambda^2-1\right| \sum_{i=1}^N \int_{\R^3} |\nabla u_i| \,
    dx + \left|\lambda^5-1\right| \sum_{i=1}^N \int_{\R^3}\int_{\R^3}
    \frac{u_i(x)u_i(y)}{8 \pi |x-y|} \, dx \, dy + g(R).
  \end{multline}
  In the limit $R \to \infty$, this yields
  $e(m')\leq \widetilde E_\infty(u^\lambda) \le e(m) (1 + C |m-m'|)$
  for a constant $C > 0$ that depends only on $m_0,\,m_1$.  Since
  $m,m'$ are arbitrary and since $e(m)$ is bounded above by the energy
  of a ball of mass $m_1$, it follows that $e$ is Lipschitz continuous
  on $[m_0,m_1]$ for all $0<m_0<m_1$.
\end{proof}

We next establish a compactness result for generalized minimizers.

\begin{lemma}[Compactness for generalized
  minimizers]\label{teocomp}
  Let $m_k$ be a sequence of positive numbers converging to some
  $m > \widetilde m_0$, where $\widetilde m_0$ is defined in
  Sec. \ref{sec-scale}, as $k\to\infty$, and let
  $(u_{k,1},\ldots,u_{k,N_k})$ be a sequence of generalized minimizers
  of $\widetilde E_\infty$ in $\widetilde{\mathcal A}_\infty(m_k)$.
  Then, up to extracting a subsequence we have that
  $N_k=N\in\mathbb N$ for all $k$, and after suitable translations
  $u_{k,i}\wto u_i$ in $BV(\R^3)$ as $k\to\infty$ for all
  $i\in \{1,\ldots,N\}$, where
  $(u_1,\ldots,u_N)\in \widetilde A_\infty(m)$ is a generalized
  minimizer of $\widetilde E_\infty$ in
  $\widetilde{\mathcal A}_\infty(m)$.
\end{lemma}

\begin{proof}
  By Theorem \ref{prodelta}, we know that $N_k\le M\in\mathbb N$ for
  all $k$ large enough. Hence, upon extraction of a subsequence we can
  asume that $N_k=N$ for all $k$, for some $N\in\mathbb N$.  For any
  $i\in\{1,\ldots,N\}$, we also have
  \begin{align}
    \sup_k \int_{\mathbb R^3} |\nabla u_{k,i}| \, dx \leq \sup_{m \in \I}
    \widetilde E_\infty(\chi_{B_{m^{1/3}}}) <\infty. 
  \end{align}
  Moreover, again by Theorem \ref{prodelta} we have
  $m_{k,i}\ge \delta_0$ and $\supp (u_{k,i}) \subset B_{D_0}(0)$,
  after suitable translations.  Hence, up to extracting a further
  subsequence, there exist $m_i\ge \delta_0$ and
  $u_i\in\widetilde A_\infty(m_i)$ such that $m_{k,i}\to m_i$ and
  $u_{k,i} \wto u_i$ in $BV(\R^3)$, as $k\to\infty$.  Passing to the
  limit in the equalities $m_k=\sum_{i=1}^N m_{k,i}$ and
  $e(m_k)=\sum_{i=1}^N e(m_{k,i})$, we obtain that
  \begin{align}
    m=\sum_{i=1}^N m_{i} \qquad {\rm and}\qquad e(m)=\sum_{i=1}^N
    e(m_{i}),    
  \end{align}
  where we used Lemma \ref{lem-elip} to establish the last
  equality. Finally, again by Lemma \ref{lem-elip} and by lower
  semicontinuity of $\widetilde E_\infty$ we have
  $e(m_i) \leq \widetilde E_\infty(u_i) \leq \liminf_{k \to \infty}
  e(m_{k,i}) = e(m_i)$, which yields the conclusion.
\end{proof}

With the two lemmas above, we are now in a position to prove the main
result of this subsection.

\begin{lemma}\label{lemI}
  The set $\I$ defined in \eqref{I} is compact.
\end{lemma}

\begin{proof}
  Since $\mathcal I$ is bounded by \cite[Theorem 3.3]{km:cpam13}, it
  is enough to prove that it is closed. Let $m_k\to m>0$, with
  $m_k\in\I$, and let $u_k \in \widetilde A_\infty(m_k)$ be such that
  $\widetilde E_\infty(u_k) = e(m_k)$ for all $k\in\mathbb N$, i.e.,
  let $u_k$ be a minimizer of the whole space problem with mass $m_k$.
  We need to prove that $m\in\I$.  By Lemma \ref{teocomp} there
  exists a minimizer $u\in\widetilde A_\infty(m)$ such that
  $u_k \wto u$ weakly in $BV(\R^3)$ and $u_k\to u$ strongly in
  $L^1(\mathbb R^3)$.  In particular, there holds
  $\widetilde E_\infty(u) = e(m)$ and hence $m\in \I$.
\end{proof}

Finally, we establish a few further properties of $e(m)$.

\begin{lemma}\label{lemII}
  Let $\lambda_m^+$ and $\lam_m^-$ be the supremum and the infimum,
  respectively, of the Lagrange multipliers in \eqref{eqELgen}, among
  all generalized minimizers of $\widetilde E_\infty$ with mass
  $m > 0$.  Then the function $e(m)$ has left and right derivatives at
  each $m\in (0,\infty)$, and
\begin{align}\label{eqder}
  \lim_{h\to 0^+}\frac{e(m+h)-e(m)}{h}=\lambda_m^-\le
  \lambda_m^+=\lim_{h\to 0^+}\frac{e(m)-e(m-h)}{h}. 
\end{align}
In particular, $e$ is a.e. differentiable and
$e'(m)=\lambda_m^-=\lambda_m^+=:\lambda_m$ for a.e. $m > 0$.
\end{lemma}
\begin{proof}
  First of all, note that for $m \leq \widetilde m_0$, where
  $\widetilde m_0$ is defined in Sec. \ref{sec-scale}, the function
  $e(m) = m f(m)$ is given via \eqref{fmball}, and the statement of
  the lemma can be verified explicitly. On the other hand, by
  definition we have $\lambda_m^-\leq \lambda_m^+$.  Fix
  $m> \widetilde m_0$ and let $(u_1,\ldots,u_N)$, with
  $u_i=\chi_{F_i}$, be a generalized minimizer of
  $\widetilde E_\infty$ with mass $m$. We first show that
\begin{align}\label{eqderuno}
  \lam_m^- \geq \limsup_{h\to 0^+}\frac{e(m+h)-e(m)}{h} 
  && %
     \text{and} 
  && %
     \lambda_m^+ \leq \liminf_{h\to 0^+}\frac{e(m)-e(m-h)}{h}.
\end{align}
Indeed, for $h>0$ let $u^h_i=\chi_{F^h_i}$ with
$F^h_i=(\frac{m+h}{m})^{1/3} F_i$, so that
$|F^h_i|=(\frac{m+h}{m})|F_i|$.  Since
$(\frac{m+h}{m})^{1/3}=1+\frac h{3m}+o(h)$, we have
\begin{align} 
  \widetilde E_\infty(u^h_i) 
  = \widetilde E_\infty(u_i) 
  & + \frac{2h}{3m} \int_{\partial F_i} \kappa(x)\
    (x\cdot\nu(x))\,d\H^2(x) \notag \\
  & + \frac {h} {12 \pi m} \int_{\partial F_i}\int_{F_i}
    \frac{(x\cdot\nu(x) )}{|x-y|}\,dy\,d\H^2(x)+o(h),
\end{align}
where $\nu(x)$ is the outward unit normal to $\partial F_i$ at point
$x$.  In view of the Euler-Lagrange equation \eqref{eqELgen}, we hence
obtain
\begin{align} %
  \widetilde E_\infty(u^h_i) - \widetilde E_\infty(u_i)
  & = \frac{\lambda h}{3m}\int_{\partial F_i} 
    (x\cdot\nu(x))\,d\H^2(x) = \lambda\left( |F^h_i|-|F_i|\right)
    + o(h), 
\end{align}
where $\lam$ is the Langrage multiplier in \eqref{eqELgen}. Passing to
the limit as $h\to 0^+$, this gives
\begin{align}\label{lello}
  \limsup_{h\to 0^+}\frac{e(m+h)-e(m)}{h} %
  \leq  \limsup_{h\to 0^+} \frac 1h \Big( \sum_{i=1}^N\widetilde
    E_\infty(u_i^h)- \sum_{i=1}^N \widetilde E_\infty(u_i) \Big) \leq
  \lambda.
\end{align}
Since \eqref{lello} holds for all generalized minimizers, this yields
the first inequality in \eqref{eqderuno}.  Following the same argument
with $h$ replaced by $-h$, and taking the limit as $h\to 0^+$, we
obtain the second inequality in \eqref{eqderuno}.

\medskip

Now, by Lemma \ref{lem-elip} the function $e(m)$ is
a.e. differentiable on $(0,\infty)$, and at the points of
differentiability we have $e'(m) = \lam_m^- = \lam_m^+ =: \lam_m$.
Hence, for any $h > 0$ there exists $m_h \in (m,m+h)$ such that $e$ is
differentiable at $m_h$ and
\begin{align}
  \frac{e(m+h)-e(m)}{h} \geq  e'(m_h)=\lambda_{m_h}\,,
\end{align}
so that
\begin{align}\label{lella}
  \liminf_{h\to 0^+}\frac{e(m+h)-e(m)}{h}\ge
  \bar\lambda:=\liminf_{h\to 0^+}\lambda_{m_h}\,. 
\end{align}
Let $h_k\to 0^+$ be a sequence such that
$\lambda_{m_{h_k}} \to \bar \lambda$ as $k \to \infty$.  If
$(u_1^k,\ldots,u_N^k)$ are generalized minimizers with mass $m_{h_k}$
then by Lemma \ref{teocomp} they converge, up to a subsequence, to a
generalized minimizer with mass $m$. In view of Lemma \ref{lemreg}, up
to another subsequence we also have that the boundaries of the
components of the generalized minimizers with mass $m_{h_k}$ converge
strongly in $C^2$ to those of the limit generalized minimizer with
mass $m$. Therefore, by \eqref{eqELgen} we have that $\bar \lambda$ is
the Lagrange multiplier associated with the limit minimizer.  It then
follows that $\bar\lambda\geq \lambda_m^-$, so that recalling
\eqref{eqderuno} and \eqref{lella} we get
\begin{align}\label{lalla}
  \lim_{h\to 0^+} \frac{e(m+h)-e(m)}{h}=\lambda_m^-.
\end{align}
This is the first equality in \eqref{eqder}. The last equality in
\eqref{eqder} follows analogously by taking the limit from the other
side.
\end{proof}

\begin{Remark}\rm
{From} the proof of Lemma \ref{lemII} it follows that $\lambda^\pm_m$
are in fact the maximum and the minimum (not only the supremum and the
infimum) of the Lagrange multipliers in \eqref{eqELgen}, i.e., that
$\lambda_m^\pm$ are attained by some generalized minimizers with mass
$m$.
\end{Remark}

\begin{corollary}\label{core}
  The function $e(m)$ is Lipschitz continuous on $[m_0,\infty)$ for
  any $m_0>0$.
\end{corollary}
\begin{proof}
  This follows from \eqref{eqder}, noticing that for all $m\geq m_0$
  there holds
\begin{align}
  -\infty < \inf_{m'\in [m_0,M]} \lambda_{m'}^-\leq
  \lambda_m^- \le\lambda_m^+  \leq \sup_{m'\in [m_0,M]}
  \lambda_{m'}^+ < +\infty,
\end{align}
where $M > 0$ is such that $\I\subset [0,M]$, and we used
\eqref{eqELgen} together with the uniform regularity from Lemma
\ref{lemreg} for the components of the generalized minimizers.
\end{proof}

\subsection{Proof of Theorem \ref{thm-fstar}}

In lieu of a complete characterization of the function $f(m)$ and the
set $\mathcal I$, we show that $f(m)$ is
continuous and attains its infimum on $\mathcal I$.  

The next result follows directly from Theorem \ref{thm-mingen},
Theorem \ref{prodelta} and \cite[Theorem 3.2]{km:cpam13}.
\begin{lemma}\label{lem-sum}
  There exists a universal constant $\delta_0 > 0$ such that for any
  $m \in (0,\infty)$ there exist $N \geq 1$ and
  $m_1,\ldots,m_{N}\in\I$ such that $m_i\ge \min \{ \delta_0, m \}$
  for all $i=1,\ldots,N$ and
\begin{align}\label{mmi}
  m=\sum_{i=1}^N m_i \qquad {\rm and}\qquad f(m)=\sum_{i=1}^N
  \frac{m_i}{m} f(m_i). 
\end{align}
\end{lemma}

Theorem \ref{thm-fstar} is a corollary of the following result.

\begin{theorem} \label{prp-fstar} %
  The function $f(m)$ is Lipschitz continuous on $[m_0,\infty)$ for
  any $m_0 > 0$. Furthermore, $f(m)$ attains its minimum, i.e.,
  \begin{align} \label{eq:Istar} %
    \mathcal I^* := \left\{ m^* \in \mathcal I \ : \ f(m^*) = \inf_{m
    \in \mathcal I} f(m) \right\} \not= \varnothing.
  \end{align}
  Furthermore, we have $f(m) \geq f^*$ for all $m > 0$ and
  \begin{align}\label{limf}
    \lim_{m\to 0}f(m)=\infty, 
    && 
       \lim_{m \to \infty} f(m) = f^*,
    && 
       \lim_{m \to \infty} \| f' \|_{L^\infty(m, \infty)} = 0. 
  \end{align}
\end{theorem}
\begin{proof}
  Since $f(m)=e(m)/m$, the Lipschitz continuity of $f(m)$ follows from
  Corollary \ref{core}.  By the continuity of $f(m)$ and since $\II$
  is compact, it then follows that there exists a (possibly
  non-unique) minimizer $m^* > 0$ of $f(m)$ over $\I$.

  \smallskip

  Turning to \eqref{limf}, the first statement there follows from
  \eqref{fmball}. Let now $u^*=\chi_{F^*} \in \AA_\infty(m^*)$ be a
  minimizer of $\widetilde E_\infty$ with $m = m^*$ for some
  $m^* \in \II^*$.  Given $k\in\mathbb N$, we can consider $k$ copies
  of $F^*$ at sufficiently large distance as a test configuration. We
  hence get $f(k m^*) \leq f(m^*)$ for any $k \in \mathbb N$, which
  implies $f(m^*) \geq \liminf_{m \to \infty}f(m)$. On the other hand,
  since $f(m^*)\leq f(m)$ for all $m\in\I$, by Lemma \ref{lem-sum} we
  obtain
  \begin{align}
    f(m)=\sum_{i=1}^N \frac{m_i}{m} f(m_i)\geq \sum_{i=1}^N
    \frac{m_i}{m} f(m^*)=f(m^*) \qquad \forall\,m>0, 
  \end{align}
  which gives the second identity in \eqref{limf}. Finally, by
  Corollary \ref{core} we have
  \begin{align}
    \lim_{m\to \infty}|f'(m)|=\lim_{m\to
    \infty}\left|\frac{e'(m)m-e(m)}{m^2}\right| \leq \lim_{m\to 
    \infty}\frac{f(m^*) + 2 \|e'\|_{L^\infty(m^*,\infty)}}{m}=0,
  \end{align}
  which yields the third identity in \eqref{limf}.
\end{proof}

\section{Proof of Theorems \ref{thm-main} and \ref{thm-uniform}}
\label{sec-main-proof}

\subsection{Compactness and lower bound}
\label{sec-lower}

In this section, we present the proof of the lower bound part of the
$\Gam$-limit in Theorem \ref{thm-main}:
\begin{proposition}[Compactness and lower bound] \label{prp-lower} %
  Let $(u_\eps) \in \AA_\eps$, let $\mu_\eps$ be given by
  \eqref{mueps}, let $v_\eps$ be given by \eqref{veps} and suppose
  that
  \begin{align} \label{EepC} \limsup_{\eps \to 0} \eps^{-4/3}
    E_\eps(u_\eps) < \infty.
  \end{align}
  Then the following holds:
  \begin{enumerate}
  \item[i)] There exists $\mu \in \MM^+(\TTT) \cap \mathcal H'$ and
    $v \in \mathcal H$ such that upon extraction of subsequences we
    have $\mu_\eps \rightharpoonup \mu$ in $\MM(\TTT)$ and
    $v_\eps \wto v$ in $\mathcal H$. Furthermore,
    \begin{align}
      \label{vdd}
      - \Delta v = \mu - \lam && \text{in $\DD'(\TTT)$.}
    \end{align}
  \item[ii)] The limit measure satisfies
    \begin{align}
      \label{EepsE0}
      E_0(\mu) \leq \liminf_{\eps \to 0} \eps^{-4/3} E_\eps(u_\eps).
    \end{align}
  \end{enumerate}
\end{proposition}

\begin{proof}
  The proof proceeds via a sequence of 4 steps. 

  \medskip

  \noindent \emph{Step 1: Compactness.} Since
  $\int_\TTT d\mu_\eps = \lam$, it follows that there is
  $\mu \in \MM^+(\TTT)$ with $\int_\TTT d\mu = \lam$ and a subsequence
  such that $\mu_\eps \rightharpoonup \mu$ in
  $\MM(\TTT)$. Furthermore, from \eqref{EepC} we have the uniform
  bound
  \begin{align} \label{id-veps} %
    \frac12 \int_{\TTT} |\nabla v_\eps|^2 \, dx = \frac12 \int_\TTT
    \int_\TTT G(x-y) \, d\mu_\eps(x) \, d\mu_\eps(y) \leq \eps^{-4/3}
    E_\eps(u_\eps) \leq C.
  \end{align}
  By the definition of the potential, we also have
  $\int_\TTT v_\eps \, dx = 0$. Upon extraction of a further
  subsequence, we hence get $v_\eps \wto v$ in $\HH$. Since
  $\mu_\eps \rightharpoonup \mu$ in $\MM(\TTT)$ and since the
  convolution of $G$ with a continuous function is again continuous,
  we also have
  \begin{align}
    \int_\TTT \left( \int_\TTT G(x-y) \phi(x) \, dx \right) d\mu_\eps(y) \to
    \int_\TTT \left( \int_\TTT G(x-y) \phi(x) \, dx \right) d\mu(y) 
    && \forall \phi \in \DD(\TTT).
  \end{align}
  This yields by Fubini-Tonelli theorem and uniqueness of the
  distributional limit that
  \begin{align}
    v(x) = \int G(x-y) \, d\mu(y) && \text{for a.e. $x \in \TTT$}.
  \end{align}
  Furthermore, since $v_\eps$ satisfies
  \begin{align} \label{vdd-eps} %
    - \Delta v_\eps = \mu_\eps - \lam && \text{in $\DD'(\TTT)$,}
  \end{align}
  taking the distributional limit, it follows that $v$ satisfies
  \eqref{vdd}. In particular, \eqref{vdd} implies that $\mu$ defines a
  bounded functional on $\mathcal H$, i.e. $\mu \in \mathcal H'$.

  \medskip

  \noindent \emph{Step 2: Decomposition of the energy into near field
    and far field contributions.}  We split the nonlocal interaction
  into a far-field and a near-field component. For $\rho \in (0,1)$
  and $x \in \TTT$, let $\eta_\rho (x) := \eta(|x| / \rho)$, where
  $\eta \in C^\infty(\R)$ is a monotonically increasing function such
  that $\eta(t) = 0$ for $t \leq \frac 12$ and $\eta(t) = 1$ for
  $t \geq 1$.  The far-field part $G_\rho$ and the near-field part
  $H_\rho$ of the kernel $G$ are then given by
  \begin{align} \label{GHrho} %
    G_\rho(x) = \eta_\rho(x) G(x), && H_\rho := G - G_\rho.
  \end{align}
  For any $u \in \mathcal A_\eps$, we decompose the energy accordingly
  as $E_\eps = E_\eps^{(1)} + E_\eps^{(2)}$, where
  \begin{align}     \label{Eeps12} %
    \begin{aligned}
      \eps^{-4/3} E_\eps^{(1)} (u) %
      & = \frac12 \eps^{-4/3} \int_\TTT \int_\TTT G_\rho(x-y) u(x)
      u(y) \, dx \, dy \\
      \eps^{-4/3} E_\eps^{(2)} (u) %
      & = \eps^{-1/3} \int_{\TTT} |\nabla u| \, dx + \frac 12
      \eps^{-4/3} \int_\TTT \int_\TTT H_\rho(x-y) u(x) u(y) \, dx \,
      dy
    \end{aligned}
  \end{align}
  In the rescaled variables, the far field part $E_\eps^{(1)}$ of the
  energy can also be expressed as
  \begin{align} \label{E1-rescaled} %
    \eps^{-4/3} E_\eps^{(1)} (u) %
    &= \frac12 \int_\TTT \int_\TTT G_\rho(x-y) \, d\mu_\eps(x) \,
      d\mu_\eps(y),
  \end{align}
  where $\mu_\eps$ is given by \eqref{mueps}. For the near field part
  $E_\eps^{(2)}$ of the energy, we set $\ell_\eps := \eps^{-1/3}$ and
  define $\t u : \TTT_{\ell_\eps} \to \R$ by
  \begin{align}
    \t u (x) :=  u(x / \ell_\eps),
  \end{align}
  where $\TTT_{\ell_\eps}$ is a torus with sidelength $\ell_\eps$
  (cf. Sec. \ref{sec-scale}). In the rescaled variables, we get
  \begin{align}
    \label{e2eps}
    \eps^{-4/3} E_\eps^{(2)} (u) %
    &= \eps^{1/3} \left( \int_{\TTT_{\ell_\eps}} |\nabla \t u| d
      x + \frac 12 \int_{\TTT_{\ell_\eps}} 
      \int_{\TTT_{\ell_\eps}} \eps^{1/3} H_\rho(\eps^{1/3}(x- y))
      \t u(x) \t u(y) \, d x \, d y 
      \right).
  \end{align}

  \medskip

  \noindent \emph{Step 3: Passage to the limit: the near field part.}
  Our strategy for the proof of the lower bound for \eqref{e2eps} is
  to compare $E_\eps^{(2)}$ with the whole space energy treated in
  Section \ref{sec-general} and use the results of this section. We
  claim that
  \begin{align} \label{est-near} %
    \liminf_{\eps \to 0} \eps^{-4/3} E_\eps^{(2)} (u) \geq (1 - c
    \rho) \lam f^*, %
  \end{align}
  for some universal constant $c > 0$. 

  \medskip
  
  Let $\Gam(x) := \frac 1{4\pi |x|}$, $x \in \R^3$, be the Newtonian
  potential in $\R^3$ and let $\Gam^\#(x) := \frac 1{4\pi |x|}$,
  $x \in \TTT$, be the restriction of $\Gamma(x)$ to the unit torus.
  We also define the corresponding truncated Newtonian potential
  $\Gam_\rho^\# : \TTT \to \R$ by
  \begin{align}
    \Gam_\rho^\#(x) := (1 - \eta_\rho(x)) \Gam^\#(x).
  \end{align}
  By a standard result, we have
  \begin{align}\label{decompG}
    G(x) = \Gam^\#(x) + R(x), \qquad x \in \TTT,
  \end{align}
  for some $R \in \mathrm{Lip}(\TTT)$.  Hence
  \begin{multline}
    \label{Hrhomin}
    H_\rho(x) %
    = (1 - \eta_\rho(x)) G(x) \geq (1 - \eta_\rho(x)) (\Gam^\#(x)
    - \nco{R}) \\
    \geq (1 - \eta_\rho(x)) \Gam^\#(x) (1 - 4 \pi \rho \nco{R} ) = (1
    - c \rho) \Gam_\rho^\#(x),
  \end{multline}
  where $c = 4\pi \nco{R}$. Inserting this estimate into
  \eqref{e2eps}, for $c \rho < 1$ we arrive at
  \begin{align} \label{updo-asin} %
    \frac{\eps^{-4/3} E_\eps^{(2)} (u) }{1 - c \rho} %
    & \geq \eps^{1/3} \left( \int_{\TTT_{\ell_\eps}} |\nabla \t u| \,
      dx + \frac 12 \int_{\TTT_{\ell_\eps}} \int_{\TTT_{\ell_\eps}}
      \eps^{1/3} \Gam_\rho^\#(\eps^{1/3}(x-y))
      \t u(x) \t u(y) \, dx \, dy \right) \notag \\
    & = \eps^{1/3} \left( \int_{\TTT_{\ell_\eps}} |\nabla \t u| \, dx
      + \int_{\TTT_{\ell_\eps}} \int_{\TTT_{\ell_\eps}} { (1 -
      \eta_\rho(\eps^{1/3} (x - y) )) \over 8 \pi |x-y| } \, \t u(x)
      \t u(y) \, dx \, dy \right).
  \end{align}

  Next we want to pass to a whole space situation by extending the
  function $\t u$ periodically to the whole of $\R^3$ and then
  truncating it by zero outside one period. We claim that after a
  suitable translation there is no concentration of the periodic
  extension of $\t u$, still denoted by $\t u$ for simplicity, on the
  boundary of a cube
  $Q_{\ell_\eps} := (-\frac12 \ell_\eps , \frac12 \ell_\eps)^3$. More
  precisely, we claim that
  \begin{align} \label{tr-good} %
    \int_{\p Q_{\ell_\eps}} \t u(x - x^*) \, d \H^2(x) \leq 6 \lam,
  \end{align}
  for some $x^* \in Q_{\ell_\eps}$. Indeed, by Fubini's theorem we
    have
    \begin{align}
      \label{fubslice}
      \lambda \ell_\eps = \int_{Q_{\ell_\eps}} \tilde u \, dx =
      \int_{-\frac12 \ell_\eps}^{\frac12 \ell_\eps} \mathcal H^2(\{
      u(x) = 1 \} \cap \{x \cdot e_1 = t\}) \, dt, 
    \end{align}
    where $e_1$ is the unit vector in the first coordinate direction.
    This yields existence of
    $x_1^* \in (-\frac12 \ell_\eps , \frac12 \ell_\eps)$ such that
    $\mathcal H^2(\{ u(x) = 1 \} \cap \{x \cdot e_1 = x_1^* \}) \leq
    \lambda$.
    Repeating this argument in the other two coordinate directions and
    taking advantage of periodicity of $\t u$, we obtain existence of
    $x^* \in Q_{\ell_\eps}$ such that \eqref{tr-good} holds.

    Now we set
    \begin{align} \label{def-tu} 
      \hat u(x) := %
      \left \{
    \begin{array}{ll}
      \t u(x - x^*) \qquad& x \in \overline Q_{\ell_\eps}, \\
      0 \qquad& x \in \R^3 \backslash \overline Q_{\ell_\eps}. \\
    \end{array}
    \right.
  \end{align}
  We also introduce the truncated Newtonian potential on $\R^3$ by
  \begin{align}\label{gammarho}
    \Gam_\rho(x) :=  {1 - \eta_\rho(x) \over 4 \pi |x|}, \qquad x
    \in \R^3.
  \end{align}
  By \eqref{tr-good}, the additional interfacial energy due to the
  extension \eqref{def-tu} is controlled:
  \begin{align} \label{low-good} %
    \int_{\TTT_{\ell_\eps}} |\nabla \t u| \, dx = \int_{\R^3} |\nabla
    \hat u| \, dx - \int_{\p Q_{\ell_\eps}} \hat u \, dx \geq
    \int_{\R^3} |\nabla \hat u| \, dx - 6 \lam .
  \end{align}
  We hence get from \eqref{updo-asin}:
  \begin{multline}
    \frac{\eps^{-4/3} E_\eps^{(2)} (u) }{1 - c \rho} %
    \geq \eps^{1/3} \left( \int_{\R^3} |\nabla \hat u| \, dx + \frac
      12 \int_{\R^3} \int_{\R^3} \Gam_{\eps^{-1/3}
        \rho}(x-y) \hat u(x) \hat u(y) \, dx \, dy - 6 \lam \right) \\
    \geq \frac \lam{\int_{\R^3} \hat u \, dx} \left( \int_{\R^3}
      |\nabla \hat u| \, dx + \frac 12 \int_{\R^3} \int_{\R^3} \Gam_{
        \rho_0}(x-y) \hat u(x) \hat u(y) \, dx \, dy \right) - 6 \lam
    \eps^{1/3},
  \end{multline}
  for any $\rho_0 > 0$, provided that $\eps$ is sufficiently small
  (depending on $\rho_0$). By Corollary \ref{cor-ER} and Theorem
  \ref{prp-fstar}, the first term on the right hand side is bounded
  below by $\lam f^*$ as soon as $\rho_0 \geq R_1$.  Therefore,
  passing to the limit as $\eps \to 0$, we obtain \eqref{est-near}.

  \medskip

  \noindent \emph{Step 4: Passage to the limit: the far field part.}
  Passing to the limit $\mu_\eps \rightharpoonup \mu$ in
  $\mathcal M(\TTT)$, for the far field part of the energy we obtain
  \begin{align}
    \label{e1rho}
    \lim_{\eps \to 0} \eps^{-4/3} E_\eps^{(1)} (u_\eps)%
    &= \frac12 \int_{\TTT} \int_{\TTT} G_\rho(x-y) \, d\mu(x) \, d \mu(y).
  \end{align}
  At the same time, by \eqref{app-eq} in Lemma \ref{lemhans} in the
  appendix the set $\{(x, y) \in \TTT : x = y \}$ is negligible with
  respect to the product measure $\mu \otimes \mu$ on
  $\TTT \times \TTT$.  Therefore, since
  $G_\rho(x - y) \nearrow G(x - y)$ as $\rho \to 0$ for all
  $x \not = y$, by the monotone convergence theorem the right-hand
  side of \eqref{e1rho} converges to
  $\int_{\TTT} \int_{\TTT} G(x - y) \, d \mu(x) \, d \mu(y)$. Finally,
  the lower bound in \eqref{EepsE0} is recovered by combining this
  result with the limit of \eqref{est-near} as $\rho \to 0$.
\end{proof}

\subsection{Upper bound construction}
\label{sec-upper}

We next give the proof of the upper bound in Theorem \ref{thm-main}:
\begin{proposition}[Upper bound construction] \label{prp-upper} %
  For any $\mu \in \mathcal M^+(\TTT) \cap \mathcal H'$ with
  $\int_\TTT d \mu = \lambda$, there exists a sequence
  $(u_\eps) \in \mathcal A_\eps$ such that
    \begin{align}
      \mu_\eps \rightharpoonup \mu \quad \text{ in
      $\MM(\TTT)$} 
      && \text{and} && v_\eps \rightharpoonup v \qquad \text{ in
                       $\mathcal H$}, 
    \end{align}
    as $\eps \to 0$, where $\mu_\eps$, $v_\eps$ and $v$ are defined in
    \eqref{mueps}, \eqref{veps} and \eqref{v}, respectively, and
    \begin{align}
      \label{amost-en}
      \limsup_{\eps \to 0} \eps^{-4/3} E_\eps(u_\eps) \leq E_0(\mu).
    \end{align}
\end{proposition}
\begin{proof}
  We first note that the limit energy is continuos with respect to
  convolutions. In particular, we may assume without loss of
  generality that $d \mu(x) = g(x) dx$ for some
  $g \in C^\infty(\TTT)$, and that there exist $C \geq c > 0$ such
  that
  \begin{align}\label{cmuepsC}
    c \leq g(x) \leq C &&\text{for all $x \in \TTT$.}
  \end{align}

  We proceed now to the construction of the recovery sequence. For
  $\delta > 0$, we partition $\TTT$ into cubes $Q_i^\delta$ with
  sidelength $\delta$.  Let $u^* \in BV(\TTT_{\ell_\eps}; \{0, 1\})$,
  where $\ell_\eps = \eps^{-1/3}$, be a minimizer of
  $\widetilde E_\infty$ over $\widetilde {\mathcal A}_\infty(m)$ with
  $m = m^* \in \II^*$ (cf. Theorem \ref{prp-fstar}), suitably
  translated, restricted to a cube with sidelength $\ell_\eps$ and
  then trivially extended to $\TTT_{\ell_\eps}$ (the latter is
  possible without modifying either the mass or the perimeter by
  Theorem \ref{prodelta} for universally small $\eps$). For a given
  set of centers $a_{\eps,\delta}^{(j)}$,
  $j = 1, \ldots, N_{\eps,\delta}$, and a given set of scaling factors
  $\theta_{\eps,\delta}^{(j)} \in [1, \infty)$, we define
  $u_{\eps,\delta} : \TTT \to \R$ by
  \begin{align} \label{muepsstar} %
    u_{\eps,\delta}(x) := \sum_{j=1}^{N_{\eps,\delta}} u^* \left(
    \theta_{\eps,\delta}^{(j)} \eps^{-1/3} (x - a_{\eps,\delta}^{(j)})
    \right) &&\text{for $x \in \TTT$},
   \end{align}
   as the sum of $N_{\eps,\delta}$ suitably rescaled minimizers of
   $\widetilde E_\infty(u)/\int_{\R^3} u \, dx$.  Note that \linebreak
   $\int_{\TTT_{\ell_\eps}} u^* (\eps^{-1/3}x) \, dx = \eps m^*$.  To
   decide on the placement of $a_{\eps,\delta}^{(j)}$, we denote the
   number of the centers in each cube as $N_{\eps,\delta}^{(i)}$,
   i.e.,
  \begin{align}
    N_{\eps,\delta}^{(i)}  := \# \Big \{ j \in \{ 1, \ldots,
    N_{\eps,\delta} \} \ : \ a_{\eps,\delta}^{(j)} \in 
    Q_i^\delta \Big \} .
  \end{align}
  With this notation we have
  $N_{\eps,\delta} = \sum_i N_{\eps,\delta}^{(i)}$, provided that
  $\supp (u_{\eps,\delta}) \cap \partial Q_i^\delta = \varnothing$ for
  all $i$.  The measure $\mu$ is then locally approximated in every
  cube $Q_i^{\delta}$ by ``droplets'' uniformly distributed throughout
  each cube. Namely, we set
  \begin{align}
    N_{\eps,\delta}^{(i)}  \ = \ \left\lceil
    \frac{\mu(Q_i^\delta)}{\eps^{1/3} 
    m^*} \right\rceil, 
  \end{align}
  and choose $a_{\eps,\delta}^{(j)}$ so that
  \begin{align}
    \label{deps}
    K \eps^{1/9} \leq d_{\eps,\delta} \leq K' \eps^{1/9},  
  \end{align}
  where
  $d_{\eps,\delta} := \min_{i \not= j} |a_{\eps,\delta}^{(j)} -
  a_{\eps,\delta}^{(i)}|$
  is the minimal distance between the centers, for some $K' > K > 0$
  depending only on $\mu$. We also set
  \begin{align}
    \theta_{\eps,\delta}^{(j)} := \left( {\eps^{1/3} m^*
    N_{\eps,\delta}^{(i)} \over
    \mu(Q_i^\delta)} \right)^{1/3} \qquad
    \text{if} \ a_{\eps,\delta}^{(j)} \in
    Q^\delta_i.
  \end{align}
  Then, if $\eps$ is sufficiently small depending only on $\delta$ and
  $\mu$, we find that $u_{\eps,\delta} \in \mathcal A_\eps$ for $\eps$
  sufficiently small depending only on $\delta$ and $\mu$.
  
  Finally, we define the measure $\mu_{\eps,\delta}$ associated with
  the test function $u_{\eps,\delta}$ constructed above,
  $d \mu_{\eps,\delta}(x) := \eps^{-{2/3}} u_{\eps,\delta}(x) \, dx$,
  as in \eqref{mueps}
  and choose a sequence of $\delta \to 0$. Choosing a suitable
  sequence of $\eps = \eps_{\delta} \to 0$, we have
  $\mu_{\eps_\delta,\delta} \rightharpoonup \mu$ in
  $\mathcal M(\TTT)$.  For simplicity of notation, in the following we
  will suppress the $\delta$-dependence, e.g., we will simply write
  $u_\eps$ instead of $u_{\eps_{\delta},\delta}$, etc.

\medskip
  
It remains to prove \eqref{amost-en}. As in the proof of the lower
bound, for a given $\rho \in (0,1)$ we split the kernel $G$ into the
far field part $G_\rho$ and the near field part $H_\rho$. Decomposing
the energy into the two parts in \eqref{Eeps12} and using
\eqref{E1-rescaled}, we have
  \begin{align}
    \eps^{-4/3} E_\eps^{(1)} (u_\eps) %
    &= \frac12 \int_\TTT \int_{\TTT} G_\rho(x - y) \, d \mu_\eps(x) \, d
      \mu_\eps(y). \label{Eeps1sup} 
  \end{align}
  Since $\mu_\eps \rightharpoonup \mu$ in $\MM(\TTT)$, we can pass to
  the limit $\eps \to 0$ in \eqref{Eeps1sup}. Then, since the limit
  measure $\mu$ belongs to $\mathcal H'$, by the monotone convergence
  theorem we recover the full Coulombic part of the limit energy $E_0$
  in \eqref{E0} in the limit $\rho \to 0$.

  \medskip

  For the estimate of the near field part of the energy, we observe
  that 
  \begin{align}
    \label{Hrhomax}
    H_\rho(x) \leq (1 + c \rho) \Gamma_\rho^\#(x), 
  \end{align}
  for some universal $c > 0$ (cf. the estimates in
  \eqref{Hrhomin}). With this estimate, we get
  \begin{align}
    \eps^{-4/3} E_\eps^{(2)}(u_\eps) \ %
    & \leq \eps^{-1/3} \int_{\TTT} |\nabla u_\eps| \, dx + \frac12
      \eps^{-4/3} (1 + c \rho) \int_\TTT \int_\TTT
      \Gamma^\#_\rho(x - y)  u_\eps(x)
      u_\eps(y) \, dx \, dy \notag \\ 
    &\leq \eps^{-1/3} \int_{\TTT} |\nabla u_\eps| \, dx + \eps^{-4/3}
      (1 + c \rho)
      \int_\TTT \int_{B_{\frac12 d_\eps}(x)} \frac{u_\eps(x)
      u_\eps(y)}{8\pi|x-y|} \, dy \, dx \notag \\ 
    &\qquad + \eps^{-4/3}
      (1 + c \rho) \int_\TTT \int_{B_\rho(x) \backslash B_{\frac12 d_\eps}(x)}
      \frac{u_\eps(x) u_\eps(y)}{8\pi|x-y|} \, dy \,
      dx. \label{linup-2}
\end{align}
By the optimality of $u^*$ and the fact that all
$\theta_{\eps,\delta}^{(j)} \geq 1$, we hence get
\begin{align}
  \eps^{-1/3} \int_{\TTT} |\nabla u_\eps| \, dx +  (1 + c
  \rho) \eps^{-4/3} 
  \int_\TTT \int_{B_{\frac12 d_\eps}(x)} \frac{u_\eps(x) 
  u_\eps(y)}{8\pi|x-y|} \, dy \, dx 
  \leq (1 + c \rho) (\lambda + o_\eps(1)) f^*,
\end{align}
where the $o_\eps(1)$ term can be made to vanish in the limit by
choosing $\eps_\delta$ small enough for each $\delta$ to ensure that
all $\theta_{\eps,\delta}^{(j)} \to 1$. Since we can choose $\rho > 0$
arbitrary, this recovers the first term in the limit energy $E_0$ in
\eqref{E0}.

\medskip

It hence remains to estimate the last term in \eqref{linup-2}. We
first note that
  \begin{align}
    \eps^{-4/3} \int_\TTT \int_{B_\rho(x) \backslash
    B_{\frac12 d_\eps}(x)} \frac{u_\eps(x) 
    u_\eps(y)}{|x-y|} \, dy \, dx %
    \leq \ \lam \, \sup_{x \in \TTT} \int_{B_\rho(x)
    \backslash B_{\frac12 d_\eps}(x)}
    \frac{d\mu_\eps(y)}{|x-y|}  .
  \end{align}
  To control the last term, for any given $x \in \TTT$ we introduce a
  family of dyadic balls $B_k := B_{2^{-k} \rho}(x)$,
  $k = 0, 1, \ldots$.  By \eqref{deps}, we have
  $B_\rho(x) \backslash B_{\frac12 d_\eps}(x) \subset
  \bigcup_{k=0}^{K_\eps} B_k \backslash B_{k+1}$
  for
  $K_\eps := \lceil \log_2 (\rho / d_\eps ) \rceil \leq 1 + \log_2
  (\rho / d_\eps)$,
  or, equivalently, $2^{-K_\eps} \rho \geq {d_\eps \over 2}$, provided
  that $\eps$ is sufficiently small depending only on $\delta$ and
  $\mu$. Therefore, with our construction we have
  $\mu_\eps(B_k) \leq 2^{-3 k} C \rho^3$ for some $C > 0$ depending
  only on $\mu$ and all $0 \leq k \leq K_\eps$. This yields
  \begin{multline}
    \label{vsup}
    \quad \sup_{x \in \TTT} \int_{B_\rho(x) \backslash B_{\frac12
        d_\eps}(x)} \frac{d\mu_\eps(y)}{|x - y|} \leq
    \sum_{k=0}^{K_\eps} \int_{B_k \backslash B_{k+1}}
    \frac{d\mu_\eps(y)}{|x - y|} \\
    \leq \sum_{k=0}^{K_\eps} \frac{2^{k +1} \mu_\eps(B_k)}{\rho} %
    \leq \sum_{k=0}^{K_\eps} \frac{2 C \rho^2}{4^k} %
    \leq {8 C \rho^2 \over 3}. \quad
    \end{multline}
    Since we can choose $\rho > 0$ arbitrarily small, this concludes the
  proof.
\end{proof}

\begin{Remark}
  \rm{We note that the construction in Proposition \ref{prp-upper}
    still yields, upon extraction of a subsequence, a recovery
    sequence for a given sequence of $\eps = \eps_n \to 0$.}
\end{Remark}

\subsection{Equidistribution of energy} 
\label{sec-proof-uniform}

We now prove Theorem \ref{thm-uniform}. First, we observe that
\begin{align}
  d \nu_\eps = \eps^{-1/3} |\nabla u_\eps| \, dx + \frac 12 v_\eps d
  \mu_\eps, 
\end{align}
where $\mu_\eps$ is defined in \eqref{mueps}. We claim that the
following lower bound for measures $\nu_\eps$, given $\bar x \in \TTT$
and $\delta \in (0, 1)$, holds true:
\begin{align} \label{EQliminf} %
  \liminf_{\eps \to 0} \nu_\eps(B_\delta(\bar x)) \ge
  |B_\delta(\bar x)| \lambda f^*.
\end{align}
As in \eqref{GHrho}, we split $G$ into the far field part $G_\rho$ and
the near field part $H_\rho$, for some fixed $\rho \in
(0,\delta)$. Since $\supp (H_\rho) \subset B_\delta(0)$, we obtain
\begin{align}
  \label{EQ1}
  \nu_\eps(B_\delta(\bar x)) 
  & = \eps^{-1/3}
    \int_{B_\delta(\bar x)} |\nabla u_\eps| \, dx + \frac12 \eps^{-4/3}
    \int_{B_\delta(\bar x)} \int_{B_\delta(x)} H_\rho(x - y) u_\eps(x)
    u_\eps(y) \, d y \, d x \notag \\
  & + \frac12 \int_{B_\delta(\bar x)} \int_\TTT G_\rho(x - y) \, d
    \mu_\eps(y) \, d \mu_\eps(x).
\end{align}
Then, since $G_\rho$ is smooth and $\mu_\eps(\TTT) = \lambda$, by
Corollary \ref{cor-conv} the integral
$\int_\TTT G_\rho(x - y) \, d \mu_\eps(y)$ converges to
$\lambda \int_\TTT G_\rho(y) \, dy$ uniformly in $x \in \TTT$ as
$\eps \to 0$. At the same time, by the definition of $G$ and
\eqref{Hrhomax} we have
$0 = \int_\TTT G(y) \, dy = \int_\TTT G_\rho(y) \, dy + \int_\TTT
H_\rho(y) \, dy \leq \int_\TTT G_\rho(y) \, dy + C \rho^2$
for some universal $C > 0$. Hence, we get
\begin{align}
  \label{EQ2}
  \nu_\eps(B_\delta(\bar x))
  & \geq \eps^{-1/3} \int_{B_\delta(\bar x)}
    |\nabla u_\eps| \, dx + \frac12\eps^{-4/3} \int_{B_\delta(\bar x)}
    \int_{B_\delta(x)} H_\rho(x - y) u_\eps(x) u_\eps(y) \, d y \, d x
    - C \lambda \rho^2,
\end{align}
for $\eps$ sufficiently small and $C > 0$ universal.

We now identify $u_\eps$ with its periodic extension to the whole of
$\mathbb R^3$. By Fubini's theorem, for a given
$\delta' \in (0, \delta)$, there is
$t = t_{\delta',\delta} \in (\delta', \delta)$ such that
\begin{align}
  \int_{\partial B_t(\bar x)} u_\eps(x) \, d \mathcal H^2(x) \leq {1
  \over \delta - \delta'} \int_{\delta'}^\delta \left(
  \int_{\partial B_s(\bar x)} u_\eps(x) \, d \mathcal H^2(x) \right) ds
  = {1 \over \delta - \delta'} \int_{B_\delta(\bar x) \backslash
  B_{\delta'}(\bar x)} u_\eps \, dx.
\end{align}
We then define $\tilde u_\eps \in BV(\mathbb R^3; \{0, 1\})$ by
$\tilde u_\eps = u_\eps \chi_{B_t(\bar x)}$. Recalling again Corollary
\ref{cor-conv}, we obtain
\begin{align} \label{Pddd} %
  \int_{\R^3} |\nabla \tilde u_\eps| \, dx %
  & = \int_{B_t(\bar x)} |\nabla u_\eps| \, dx + \int_{\partial
    B_t(\bar x)} u_\eps(x) \, d\HH^2(x) \notag \\ %
  &\leq \int_{B_\delta(\bar x)} |\nabla u_\eps| \, dx + C \lambda
    \delta^2 \eps^{2/3},
\end{align}
for some universal $C > 0$, provided that $\eps$ is sufficiently
small. We note that $\tilde u_\eps(x) \leq u_\eps(x)$ for every
$x \in \mathbb R^3$. Furthermore, for sufficiently small $\delta$ we
have $H_\rho \geq 0$ and
\begin{align} \label{esto}
H_\rho(x - y) \geq {(1 - c \rho) \Gam(x-y)}
\qquad \text{for all $|x - y| \leq \tfrac12 \rho$,}
\end{align}
for some universal $c > 0$ (where $\Gamma$ is the Newtonian potential
in $\R^3$, as above).
From \eqref{EQ2}, \eqref{Pddd} and \eqref{esto} we then get
\begin{align} \label{EQ3} %
  \nu_\eps(B_\delta(\bar x)) \geq \eps^{-1/3} \int_{\mathbb R^3}
  |\nabla \tilde u_\eps| \, dx + \frac{1-c\rho}{2}\eps^{-4/3}
  \int_{\mathbb R^3} \int_{B_{\rho/2}(x)} \Gamma(x-y) \tilde u_\eps(x)
  \tilde u_\eps(y) \, dy \, dx - C \lambda \rho^2,
\end{align}
for $\eps$ small enough. Letting now
$\hat u_\eps(x) := \tilde u_\eps(\eps^{1/3} x)$ be the rescaled
function which satisfies
\begin{align}
  \int_{\mathbb R^3} \hat u_\eps \, dx = \frac 1 \eps \int_{\mathbb R^3}
  \tilde u_\eps \, dx = \lambda |B_t(\bar x)|\eps^{-1/3} +o(\eps^{-1/3}),
\end{align}
for every fixed $\rho_0 > 0$ and $\eps$ sufficiently small, we get
\begin{multline}
  \label{EQ4}
  \nu_\eps(B_\delta(\bar x)) \\
  \geq (1 - c \rho) \eps^{1/3} \left( \int_{\mathbb R^3} |\nabla \hat
    u_\eps| \, dx + \frac 12\int_{\mathbb R^3} \int_{B_{\rho_0}(x)}
    \Gamma(x-y)\hat u_\eps(x) \hat u_\eps(y) \, dy \,dx \right) - C
  \lambda \rho^2
  \\
  \geq \frac{(1-2 c\rho)\lam |B_t(\bar x)|}{\int_{\R^3} \hat u_\eps \,
    dx} \left( \int_{\R^3} |\nabla \hat u_\eps| \, dx + \frac 12
  \int_{\R^3} \int_{\R^3} \Gam_{\rho_0}(x-y) \hat u_\eps(x) \hat
  u_\eps(y) \, dy \, dx \right) - C \lambda \rho^2,
\end{multline}
where $\Gam_{\rho_0}$ is defined via \eqref{gammarho}.  Recalling
Corollary \ref{cor-ER} and choosing $\rho_0\geq R_1$, we obtain
\begin{align}
  \liminf_{\eps\to 0} \nu_\eps(B_\delta(\bar x)) \geq (1-2 c\rho)
  \lambda f^* |B_t(\bar x)| - C \lambda \rho^2,  
\end{align}
which gives \eqref{EQliminf} by first letting $\rho\to 0$ and then
$\delta' \to \delta$.

\smallskip

We now prove a matching upper bound.  Notice that by the definition we
have $v_\eps(x) \geq C := - \lambda |\min_{y \in \TTT} G(y)|$ for
every $x \in \TTT$. Therefore, the negative part $\nu^-_\eps$ of
$\nu_\eps$ obeys
$\nu^-_\eps(U) = -\frac12 \int_{U \cap \{v_\eps < 0 \} } v_\eps
d\mu_\eps \leq \frac12 |C| \mu_\eps(U)$
for every open set $U \subset \TTT$.  In turn, since
$\nu_\eps(\TTT) = \lambda f^* + o_\eps(1)$ by \eqref{Eepsmin}, it
follows that the positive part $\nu^+_\eps$ of $\nu$ obeys
$\nu^+_\eps(U) = \int_{U \cap \{ v_\eps \geq 0 \} } \left( \eps^{-1/3}
  |\nabla u_\eps| \, dx + \frac12 v_\eps \, d \mu_\eps \right) \leq
\lambda f^* + \frac12 |C| \lambda + o_\eps(1)$.
Hence $|\nu_\eps| = \nu^+_\eps + \nu^-_\eps$ is uniformly bounded as
$\eps\to 0$, and up to a subsequence $\nu_\eps \rightharpoonup \nu$
for some $\nu\in \mathcal M(\TTT)$ with $\nu(\TTT) = \lambda f^*$.
Since from the lower bound \eqref{EQliminf} we have
$\nu(U) \ge \lambda f^* |U|$, it then follows that
$d \nu = \lambda f^* dx$.  Finally, in view of the uniqueness of the
limit measure, the result holds for the original sequence of
$\eps \to 0$. \qed

\section{Uniform estimates for minimizers of the rescaled energy}
\label{sec-prop}

In this section, we establish uniform estimates for the minimizers of
the rescaled problem associated with $\widetilde E_\ell$ over
$\widetilde A_\ell$ from \eqref{El} and \eqref{Al}, respectively. The
main result is a uniform bound on the modulus of the potential,
independently of the domain size $\ell$.

\medskip

Throughout this section, $F \subset \TTT_\ell$ with
$|F| = \lambda \ell$ is always taken to be such that
$\tilde u_\ell = \chi_F$ is a regular representative of a minimizer of
$\widetilde E_\ell$ over $\widetilde A_\ell$ for a given $\lambda > 0$
(for simplicity of notation, we suppress the explicit dependence of
$F$ on $\ell$ throughout this section). The estimates below are
obtained for families of minimizers $(\tilde u_{\ell_n})$ as
$\ell_n \to \infty$ and hold for all $\ell_n \geq \ell_0$, where
$\ell_0 > 0$ may depend on $\lambda$ and the choice of the family. For
simplicity of notation, we indicate this by saying that an estimate
holds for $\ell \gg 1$.

\medskip

Following \cite{rigot00,maggi} we recall the notion of
$(\Lambda,r_0)$-minimizer of the perimeter (for a different approach
that leads to the same regularity results, see \cite{goldman12cvar}).

\begin{definition}
  Given $\Lambda,r_0>0$ we say that a set $F \subset \TTT_\ell$ is a
  volume-constrained $(\Lambda,r_0)$-minimizer if
\begin{equation}\label{eqlambda}
  P(F)\leq P(F')+\Lambda |F\Delta F'| \qquad \forall\, F'\subset\TTT_\ell,
  \text{ s.t. }(F\Delta F')\subset B_{r_0}\text{ and }|F'|=|F|\,,
\end{equation}  
where $P(F)$ denotes the perimeter of the set $F$, and $B_r$ denotes a
generic ball of radius $r$ contained in $\TTT_\ell$.
\end{definition}

The following result is a consequence of the regularity theory for
minimal surfaces with volume constraint (see for instance
\cite[Chapters III--IV]{maggi}, \cite[Section 4]{rigot00}).

\begin{proposition}\label{thgn}
  Let $F\subset\TTT_\ell$ be a volume-constrained
  $(\Lambda,r_0)$-minimizer, with \linebreak
  $|F|\in \left( r_0^3, \ell^3 - r_0^3 \right)$.  Then $\partial F$ is
  of class $C^{1,1/2}$, and there exist universal constants
  $\delta > 0$ and $c>0$ such that for all $x_0 \in \overline F$ we
  have
 \begin{align} 
   \label{Lamr0dens}
   |F_0 \cap B_r(x_0)|\geq c r^3 \qquad \text{for all } 
   r\leq \min\left(r_0,\frac{\delta}{\Lambda}\right)\,,
  \end{align}
  where $F_0$ is the connected component of $F$ such that
  $x_0 \in \overline F_0$.
\end{proposition}

Let $\Gam(x) := {1 \over 4\pi |x|}, \ x \in \R^3$, be the Newtonian
potential in $\R^3$ and let
$\Gam^\#_\ell(x) := {1 \over 4\pi |x|}, \ x \in \TTT_\ell$, be the
restriction of $\Gamma(x)$ to $\TTT_\ell$.  Letting
\begin{align}
  \label{Gell}
  \widetilde G_\ell(x) := {1 \over \ell} G \left( {x \over
  \ell} \right), \qquad \quad
  x \in \TTT_\ell,
\end{align}
by \eqref{decompG} we have for all $\ell \geq 1$:
\begin{align}\label{decompGG}
  \widetilde G_\ell(x) = \Gam^\#_\ell(x) + R_\ell(x) \qquad {\rm for\ all\
  }x \in \TTT_\ell,
\end{align}
with $R_\ell \in {\rm Lip}(\TTT_\ell)$ satisfying
\begin{align}\label{eqR}
  |R_\ell(x)|\leq \frac{C}{\ell}\qquad {\rm and}
  \qquad |\nabla R_\ell(x)|\leq \frac{C}{\ell^2}
\qquad {\rm for\ all\
  }x\in\TTT_\ell\,,
\end{align}
with a universal $C>0$.

Let now 
\begin{align}
  v_F(x) := \int_F \widetilde G_\ell(x-y) \, dy, \qquad \quad
  x \in \TTT_\ell,
\end{align}
be the potential associated with $F$.  Notice that $v_F$ satisfies
\begin{align}
  -\Delta v_F = \chi_F - {\lambda \over \ell^2} \qquad \text{and}
  \qquad \int_{\TTT_\ell} v_F \, dx = 0.
\end{align}
In particular, by standard elliptic regularity
$v_F \in C^{1,\alpha}(\TTT_\ell)$ for any $\alpha \in (0, 1)$
\cite{gilbarg}, and $v_F$ is subharmonic outside $\overline F$, so
that the maximum of $v_F$ is attained in $\overline F$. Moreover, we
have the following a priori bounds for $v_F$ throughout the rest of
this section, $v_F$ always refers to the potential associated with the
minimizer $F$).

\begin{lemma}
  \label{lem-estv}
  There exists a universal constant $C > 0$ such that
  \begin{align}
  \label{estvE}
    -C \leq v_F \leq C (\lambda \ell)^{2/3},
  \end{align}
  for all $\ell \gg 1$.
\end{lemma}

\begin{proof}
  First of all, observe that $v_F(x) = v_\eps(\eps^{1/3} x)$ for
  $\eps = \ell^{-3}$, where $v_\eps$ is defined in \eqref{veps}, in
  which $\mu_\eps$ is given by \eqref{mueps} with
  $u_\eps(x) = \chi_F(\eps^{-1/3} x)$. Furthermore, by a rescaling we
  have that $u_\eps$ is a minimizer of $E_\eps$ over
  $\mathcal A_\eps$. Therefore, to establish a lower bound for $v_F$,
  it is sufficient to do so for $v_\eps$.

  Let $G_\rho$ and $H_\rho$ be as in \eqref{GHrho} (with the choice of
  $\eta$ fixed once and for all), and note that there exists a
  universal $\rho_0 > 0$ such that $H_\rho \geq 0$ for all
  $\rho \in (0, \rho_0)$ and, hence,
  \begin{align}
    \label{vepsGrho}
    v_\eps(x) \geq \int_\TTT G_\rho(x - y) \, d \mu_\eps(y).
  \end{align}
  At the same time, by Corollary \ref{cor-conv} and the boundedness of
  $|\nabla G_\rho|$ we have
  \begin{align}
    \label{Grhodx}
    \int_\TTT G_\rho(x - y) \, d \mu_\eps(y) \to \lambda
    \int_\TTT G_\rho(y) \, dy \qquad \text{uniformly in } x \in \TTT, 
  \end{align}
  as $\eps \to 0$. Notice that from the definition of $G$ we have
  $0 = \int_\TTT G(x) \, dx = \int_\TTT G_\rho(x) \, dx + \int_\TTT
  H_\rho(x) \, dx $. Therefore, by \eqref{decompG} we get
  \begin{align}
    \label{G0r2}
    -C \rho^2  \leq \int_\TTT G_\rho(x) \, dx \leq 0,
  \end{align}
  for some universal $C > 0$ and all $\rho \in (0, \rho_0)$.  Choosing
  $\rho = \min \{ \rho_0, \lambda^{-1/2} \}$, we then obtain
  $v_\eps \geq -2 C$ for all $\eps > 0$ sufficiently small.

  On the other hand, by \eqref{eqR} there exists a universal constant
  $C > 0$ such that
  \begin{align}
    v_F(x) 
    & \leq C \int_F {dy \over |x - y|} \leq C \left(
      \int_{B_R(x)} {dy \over |x - y|} + { |F \backslash B_R(x)|
      \over R} \right) \notag \\
    & \leq C (2 \pi R^2 + R^{-1} |F|),
  \end{align}
  for any $\ell \geq 1$ and $R > 0$.  The claim then follows by
  choosing $R = |F|^{1/3} = (\lambda \ell)^{1/3}$.
\end{proof}

\begin{Remark} \rm Let $\lambda_0 > 0$ and let
  $\lambda \in (0, \lambda_0)$. Since
  $v_F \geq \lambda \min_{x \in \TTT} G(x)$, it is also possible to
  obtain a lower bound on $v_F$ which depends only on $\lambda_0$, and
  not on the family of the minimizers, provided that
  $\ell \geq \ell_0$ for some $\ell_0 > 0$ depending only on
  $\lambda$. In this case all the estimates of this section still
  hold, but with constants that depend on $\lambda_0$.
\end{Remark}

We next obtain a pointwise estimate of the gradient of $v_F$ in terms
of $v_F$ itself.

\begin{lemma}\label{stimagrad}
  There exists a universal constant $C>0$ such that for every
  $\ell \gg 1$ we have
  \begin{align}\label{CCC}
    |\nabla v_F(x)|\le \frac32 \left( v_F(x)+ C \right) \,,
  \end{align}
  for any $x\in\TTT_\ell$.
\end{lemma}

\begin{proof}
  Without loss of generality we may assume that $x = 0$. Arguing as in
  the proof of Lemma \ref{lem-estv} and with the same notation, we can
  write
  \begin{align}
    \label{gradvf}
    |\nabla v_F(0)| \leq \int_F |\nabla \widetilde G_\ell(y)| \, dy = \ell \int_\TTT
    | \nabla G(y)| \, \chi_F(y \ell) \, dy = \eps^{1/3} \int_\TTT  |
    \nabla G(y)| \, d \mu_\eps(y) \notag \\
    \leq \eps^{1/3} \int_\TTT  |
    \nabla G_\rho(y)| \, d \mu_\eps(y) + \eps^{1/3} \int_\TTT  |
    \nabla H_\rho(y)| \, d \mu_\eps(y),
  \end{align}
  where we recalled that $\eps = \ell^{-3}$.  Using \eqref{decompG},
  we have 
  \begin{align}
    |\nabla H_\rho(y)| \leq (1 + c \rho) |y|^{-1} H_\rho(y) + C
    |y|^{-1} \rho^{-1} \chi_{B_\rho \backslash B_{\rho/2}}(y),
  \end{align}
  for some universal $c, C > 0$ and all $\rho \in (0, \rho_0)$.
  Substituting this into \eqref{gradvf} and recalling \eqref{mueps}
  and \eqref{Hrhomax}, we obtain
  \begin{align}
    \label{gradvFc}
    {|\nabla v_F(0)| \over 1 + c \rho} 
    & \leq \int_{\TTT \backslash
      B_{\eps^{1/3}}(0)} H_\rho(y) \, d \mu_\eps(y) + \eps^{-1/3}
      \int_{B_{\eps^{1/3}}(0)}
      |y|^{-1} H_\rho(y) \, u_\eps(y) \, dy \notag \\
    & + C \eps^{1/3} \rho^{-1} \int_{B_\rho(0) \backslash B_{\rho/2}(0)}
      |y|^{-1} \, d \mu_\eps(y) + \eps^{1/3} \int_\TTT |
      \nabla G_\rho(y)| \, d \mu_\eps(y) \\
    & \leq \int_\TTT H_\rho(y) \, d \mu_\eps(y) + C' (1 + \eps^{1/3}
      \rho^{-2} \lambda ) + \eps^{1/3} \int_\TTT | \nabla G_\rho(y)| \,
      d \mu_\eps(y), \notag
  \end{align}
  for some universal $C, C' > 0$.  Since by Corollary \ref{cor-conv}
  and the smoothness of $G_\rho$ we have
  $\int_\TTT | \nabla G_\rho(x - y)| \, d \mu_\eps(y) \to \lambda
  \int_\TTT | \nabla G_\rho(y)| \, d y$
  uniformly in $x \in \TTT$ as $\eps \to 0$, it is possible to choose
  $\eps_0 > 0$ sufficiently small independently of $x$ such that the
  last two terms in the right-hand side of \eqref{gradvFc} are bounded
  by a universal constant for all $\eps < \eps_0$. Thus, for every
  $\rho \leq 1/(2 c)$ and $\eps < \eps_0$, with $\eps_0$ depending on
  $\rho$, we have
  \begin{align}
    \frac23 |\nabla v_F(0)| \leq v_F(0) + C - \int_\TTT G_\rho(y) \, d
    \mu_\eps(y),
  \end{align}
  where we also took into account that
  $v_F(0) = \int_\TTT G_\rho(y) \, d \mu_\eps(y) + \int_\TTT H_\rho(y)
  \, d \mu_\eps(y)$.
  Finally, using \eqref{Grhodx} and \eqref{G0r2}, we obtain
  \begin{align}
    |\nabla v_F(0)| \leq  \frac32 v_F(0) + C(1 +  \lambda \rho^2),
  \end{align}
  for some universal $C > 0$ and all $\eps < \eps_0$, possibly
  decreasing the value of $\eps_0$. The proof is concluded by choosing
  $\rho \leq \lambda^{-1/2}$.
\end{proof}

\begin{corollary}\label{corgrad}
  Let $\ell \gg 1$ and let $\bar x \in \overline F$ be a global
  maximum of $v_F$.  Then
  \begin{align}
    v_F(y)\ge \frac 34\, v_F(\bar x)-\frac14 C
    \qquad\text{for all }y\in 
    B_{1/6}(\bar x)\,,    
  \end{align}
  where $C$ is as in \eqref{CCC}. Furthermore, if
  $\int_{B_r(x_0)} v_F(x) \, dx \leq C' |B_r| $ for some
  $x_0 \in \TTT_\ell$, $r \leq \frac16$ and $C' > 0$, then
  \begin{align}
    v_F(y) \le C + 2 C'
    \qquad\text{for all }y\in  B_r(x_0)\,,
  \end{align}
\end{corollary}

\begin{proof}
  Since $v_F \in C^1(\TTT_\ell)$, for any $y\in B_{1/6}(\bar x)$ there
  exists $\theta \in (0,1)$ such that with the help of \eqref{CCC} we
  have
  \begin{align}
    v_F(\bar x)-v_F(y)
    &=\nabla v_F(\theta \bar x + (1-\theta)y)\cdot (\bar x-y)
      \notag \\
    &\le \frac16 \, |\nabla v_F(\theta \bar x + (1-\theta)y)|
      \notag \\
    &\le \frac14 \, v_F(\theta \bar x + (1-\theta)y)+\frac14 C  
      \notag \\
    &\le \frac14 \, v_F(\bar x) +\frac14 C \,.
  \end{align}

  Similarly, letting $\bar y$ be a global maximum of $v_F$ in
  $\overline B_r(x_0)$ and letting $x_1 \in \overline B_r(x_0)$ be
    such that $v_F(x_1) = |B_r|^{-1} \int_{B_r(x_0)} v_F(x) \, dx$, we
    may write
  \begin{align} 
    v_F(\bar y)
    &\le v_F(\bar y)-v_F(x_1) + C'
      \notag \\
    &\le |\nabla v_F(\theta x_1 + (1-\theta)\bar y)|\, |\bar y-x_1|+
      C' 
      \notag \\
    &\le \frac12\, v_F(\theta \bar x_1 + (1-\theta)\bar y) + \frac12
      C + C'
      \notag \\
    &\le \frac 12\, v_F(\bar y) + \frac12
      C + C' \,,
\end{align}
which completes the proof.
\end{proof}

The next lemma provides a basic estimate for the variation of the
Coulombic energy under uniformly bounded perturbations.

\begin{lemma}\label{afm}
  There exists a universal constant $C>0$ such that for any
  $\ell \geq 1$ and for any $F'\subset \TTT_\ell$, with
  $F\Delta F'\subset B_r(x_0)$ for some $x_0 \in \TTT_\ell$ and
  $r > 0$, there holds
  \begin{align}
    \left| \int_{F}v_F \, dx -\int_{F'}v_{F'} \, dx 
    \right|\le
    \left( 2 \|v_F\|_{L^\infty(\TTT_\ell)} + C r^2 \right)\,
    |F\Delta F'|\,.     
  \end{align}
\end{lemma}

\begin{proof}
By direct computation, we have
\begin{align}
  \hspace{6ex} 
  & \hspace{-6ex} %
    \left| \int_{F} 
    v_F\,dx -\int_{F'}v_{F'}\,dx \right| 
    = \left|
    \int_{\TTT_\ell} \int_{\TTT_\ell} \left( \chi_F(x) \widetilde G_\ell(x - y)
    \chi_F(y) - \chi_{F'}(x) \widetilde G_\ell(x - y)
    \chi_{F'}(y) \right) \, dx \, dy \right| \notag \\
  & = \left| \int_{\TTT_\ell} \int_{\TTT_\ell} (\chi_F(x) + \chi_{F'}(x)
    ) \widetilde G_\ell(x - y)  (\chi_F(y) - \chi_{F'}(y) ) \, dx
    \, dy \right| \notag \\ 
  & \leq 2 \left| \int_{\TTT_\ell} \int_{\TTT_\ell} \chi_F(x) \widetilde G_\ell(x -
    y)  (\chi_F(y) - \chi_{F'}(y) ) \, dx \, dy \right| \notag \\ 
  & \qquad +  \left| \int_{\TTT_\ell}
    \int_{\TTT_\ell} (\chi_F(x) - \chi_{F'}(x) 
    ) \widetilde G_\ell(x - y)  (\chi_F(y) - \chi_{F'}(y) ) \, dx \, dy \right|
    \notag \\ 
  & \leq 2 \left| \int_{\TTT_\ell} v_F(y)  (\chi_F(y) - \chi_{F'}(y) )\,
    dy \right| + 2 \left| \int_{\TTT_\ell}  \int_{B_r(y)} \widetilde G_\ell(x - y)
    (\chi_F(y) - \chi_{F'}(y) )\, dx \, dy \right| \notag \\
  & \leq \left( 2 \|v_F\|_{L^\infty(\TTT_\ell)} + C r^2 \right) | F \Delta F'|,
\end{align}
for some universal $C > 0$, where we used \eqref{decompGG} and
\eqref{eqR} in the last line.
\end{proof}

Lemma \ref{afm} implies that minimizers of $\widetilde E_\ell$ are
volume constrained $(\Lambda,r_0)$-minimizers of the perimeter for
$r_0 = 1$ and $\Lambda = \| v_F \|_{L^\infty(\TTT_\ell)} + C$, with
$C > 0$ universal. In particular, by Lemma \ref{lem-estv} we get
$\Lambda \leq C(\lambda\ell)^{2/3}$, provided that $\ell \gg 1$.
Therefore, from Proposition \ref{thgn} we obtain the following result.

\begin{proposition}\label{thrigot}
  There exist universal constants $c > 0$ and $\delta>0$ such that for
  all $\ell \gg 1$ and all $x_0 \in \overline F$ there holds
  \begin{align} 
    \label{densfirst} |F_0 \cap B_r(x_0)|\geq c r^3 \qquad
    \text{for all } r\leq 
    \frac{\delta}{(\lambda\ell)^{2/3}} \,,
  \end{align}
  where $F_0$ is the connected component of $F$ such that
  $x_0 \in \overline F_0$.
\end{proposition}

We now show that the potential $v_F$ is bounded in
$L^\infty(\TTT_\ell)$ by a universal constant as $\ell \to \infty$.

\begin{theorem}[$L^\infty$-estimate on the potential]
  \label{lem-pot} %
  There exists a universal constant $C>0$ and a constant $\ell_0 > 0$
  such that for all $\ell \geq \ell_0$ we have
  \begin{align} \label{stimav} %
    \|v_F\|_{L^\infty(\TTT_\ell)} \leq C.
  \end{align}
\end{theorem}

\begin{proof}
  Observe first that by \eqref{estvE} we have $v_F \geq -C$, for some
  universal constant $C > 0$ and $\ell \gg 1$. Therefore, letting
  $V:= \displaystyle \max_{x \in \TTT_\ell} v_F(x)$, the thesis is
  equivalent to showing that
  \begin{align} \label{stimavv} %
    V \leq C,
  \end{align}
  for some universal $C>0$ and large enough $\ell$.

  We first prove \eqref{stimavv} with the constant depending only on
  $\lambda$. Partition $\TTT_\ell$ into $N$ cubes of sidelength
  $L = \ell\,N^{-1/3}$, with $N^{1/3}$ chosen to be the smallest
  integer such that
  $L \le \min \left(\frac16 c^{1/3} \lambda^{-1} \delta, \frac13
  \right)$,
  where $c$ and $\delta$ are as in \eqref{densfirst}.  Note that with
  our choice of $L$ we have
  $N \geq 216 \lambda^3 \ell^3 / (c \delta^3)$.  If $\ell$ is
  sufficiently large (depending on $\lambda$), we also have that
  $\delta(\lambda\ell)^{-2/3}\le \frac12 L \le \frac{1}{12} c^{1/3}
  \lambda^{-1} \delta$.
  In particular, any ball of radius $\delta(\lambda\ell)^{-2/3}$ can
  be inscribed into a union of $27$ adjacent cubes of the partition
  and stay at least distance $\delta(\lambda\ell)^{-2/3}$ from the
  boundary of that union.  Hence, by \eqref{densfirst} and a counting
  argument we get that at least $\tfrac78 N$ cubes do not intersect
  $F$, so that we can find disjoint balls $B_1,\ldots,B_{M}$ of radius
  $\frac12 L \leq \frac16$ not intersecting $F$, with
  $M \geq \tfrac78 N$.
  
  Recalling that $\int_{\TTT_\ell} v_F\, dx=0$ and that $v_F$ is
  bounded below by $-C$, for $\ell \gg 1$ we get
  \begin{align}
    0 = \int_{\TTT_\ell} v_F\, dx
    \geq \sum_{i=1}^{M}\int_{B_i} v_F\, dx -C \ell^3.
  \end{align}
  It follows that there exists an index $i$ such that, for some
  universal $C' > 0$, we have
  \begin{align}
    \int_{B_i} v_F\, dx\leq C M^{-1} \ell^3 \leq C' |B_i|
    \,.  
  \end{align}
  We then apply the second part of Corollary \ref{corgrad} with
  $x_0 = x_i$, where $x_i$ is the center of $B_i$, to obtain
  \begin{align}
    \label{vFBi}
    |v_F (x)|\le C \qquad \text{for all} \ x \in 
    B_i,
  \end{align}
  for some universal $C>0$.

  \smallskip

  Let now $\bar x \in\overline F$ be a global maximum of $v_F$, so
  that $v_F(\bar x)=V$, and assume that
  \begin{align}\label{eqrr}
    \mathcal H^2(F\cap\partial B_{r}(\bar x))\ge \frac19 V
    |F\cap  B_{r}(\bar x)|\qquad
    \text{for any }r\in (0,L/2)\,.
  \end{align}
  Letting $m(r):=|F\cap B_{r}(\bar x)|$, so that
  ${dm(r) \over dr} =\mathcal H^2(F\cap\partial B_r(\bar x))$ for
  a.e. $r$, \eqref{eqrr} can be written as
  \begin{align}\label{eqdiff}
    {dm(r) \over dr} \ge \frac19 V m(r) \qquad \text{for a.e. }r\in (0,L/2)\,.
  \end{align}
  Integrating \eqref{eqdiff} over $(r_0,L/2)$, we get (for a similar
  argument, see the proof of \cite[Theorem 3.3]{km:cpam13})
  \begin{align}
    m(r_0)\le m(L/2)\, e^{V(r_0 - L/2)/9}\,.    
  \end{align}
  Notice now that, as in Proposition \ref{thrigot}, from Lemma
  \ref{afm} it follows that
  \begin{align}
    m(r)\ge cr^3 \qquad \text{for all }r \leq \min \left( 1, \frac
    \delta V \right)\,.    
  \end{align}
  In particular, if $r_0 = \delta / V \le L/4$, we have
  \begin{align}\label{eqVV}
    \frac{c \delta^3}{V^3}\le m(r_0)
    \le C L^3 e^{-LV/36} \,,
  \end{align}
  for some universal constant $C>0$, which implies \eqref{stimavv}
  with the constant depending only on $\lambda$.

  On the other hand, if \eqref{eqrr} does not hold, there exists
  $r\in (0,L/2)$ such that
  \begin{align}\label{eqrg}
    \mathcal H^2(F\cap\partial B_{r}(\bar x)) < \frac19 V
    |F\cap B_{r}(\bar x)|\,.
  \end{align}
  We claim that, as in the proof of Lemma \ref{lem-ln}, if
  \eqref{stimavv} does not hold, it is convenient to move the set
  $F\cap B_{r}(\bar x)$ inside the ball $B_i$.  Indeed, we define
  $F_i:=(x_i-\bar x)+(F\cap B_{r}(\bar x))$ and
  $\hat u = \tilde u_\ell - \chi_{F\cap B_{r}(\bar x)} + \chi_{F_i}$.
  Note that by construction $F \cap B_i = \varnothing$, so $\hat u $
  is admissible. By minimality of $\tilde u_\ell$ and using
  \eqref{eqR}, \eqref{vFBi} and \eqref{eqrg}, we get
  \begin{align}
    \widetilde E_{\ell}(\tilde u_\ell)
    &\leq \widetilde E_{\ell}(\hat u)
      \notag  \\   
    & =  \widetilde E_{\ell}(\tilde u_\ell) 
      + 2 \mathcal H^2(F\cap\partial B_{r}(\bar x)) 
      +\int_{F_i}v_F\,dx - \int_{F\cap B_{r}(\bar x)}v_F\,dx
      \notag \\
    &\qquad - \int_{F_i}\int_{F\cap B_{r}(\bar x)}\widetilde G_\ell(x-y)\,dx \, dy
      + \int_{F\cap B_{r}(\bar x)}\int_{F\cap B_{r}(\bar
      x)}\widetilde G_\ell(x-y)\,dx \, dy
      \notag \\
    &<  \widetilde E_{\ell}(\tilde u_\ell) 
      + \left( \frac29 V + C \right) |F\cap B_{r}(\bar x)|
      - \int_{F\cap B_{r}(\bar x)}v_F\,dx\,,
  \end{align}
  for some universal $C>0$, provided that $\ell \gg 1$. Notice now
  that Corollary \ref{corgrad} implies that
  \begin{align}
    v_F(x)\ge \frac 34 V - C \qquad\text{for any
    }x\in B_r(\bar x),   
  \end{align}
  for a universal $C > 0$.  Hence
  \begin{align}
    0 < \left(C- \frac12 V \right)|F\cap B_{r}(\bar x)|,  
  \end{align}
  for some universal $C > 0$ and $\ell \gg 1$, which leads to a
  contradiction if $V$ is too large.

  \medskip

  Lastly, to establish \eqref{stimavv} with $C$ universal, we note
  that using \eqref{stimavv} with the constant depending on $\lambda$
  one gets that the density estimate in \eqref{densfirst} holds for
  all $r \leq r_0$ with some $r_0 > 0$ depending only on $\lambda$,
  for $\ell \gg 1$. We can then repeat the covering argument at the
  beginning of the proof with $L > 0$ universal, provided that
  $\ell \gg 1$, and obtain the conclusion by repeating the above
  argument.
\end{proof}

{From} Theorem \ref{lem-pot} and the arguments leading to Proposition
\ref{thrigot}, we obtain an improved density estimate for minimizers
of $\widetilde E_\ell$.

\begin{corollary}\label{cor-dens} %
  There exist a universal constant $c>0$ and a constant $\ell_0 > 0$
  such that for all $x_0 \in \overline F$ and all $\ell \geq \ell_0$
  we have
 \begin{align} \label{crst} %
    |F_0 \cap B_r(x_0)|\geq c r^3 \qquad \text{for all $r\leq 1$,}
  \end{align}
  where $F_0$ is the connected component of $F$ such that
  $x_0 \in \overline F_0$.
\end{corollary}

Finally, we establish a uniform diameter bound for the connected
components of the minimizers in Theorem \ref{lem-pot}.

\begin{lemma}[Diameter bound] \label{lem-diam} %
  Let $F_0$ be a connected component of $F$. Then there exists a
  universal constant $C>0$ such that
  \begin{align} \label{diamF0} %
    \mathrm{diam} \, F_0 \leq C,
  \end{align}
  for all $\ell \gg 1$.
\end{lemma}

\begin{proof}
  Assume that ${\rm diam} \, F_0 \geq 2$. Arguing as in the proof of
  Lemma \ref{stimagrad} and using its notations, for any
  $x \in \TTT_\ell$ and a universally small $\rho_0 > 0$ we have
  \begin{align}
    \label{vFGepsx}
    v_F(x) \geq \int_{F \cap B_{\eps^{-1/3} \rho/2}(x)} {dy \over 8
    \pi |x - y|} + \int_\TTT G_\rho(\eps^{1/3} x - y) \, d
    \mu_\eps(y), 
  \end{align}
  for all $\rho \in (0, \rho_0)$. Observe that by \eqref{Grhodx} and
  \eqref{G0r2} the last term in the right-hand side of \eqref{vFGepsx}
  can be bounded below by $-2 C \lambda \rho^2$, for $\ell \gg 1$ and
  $C > 0$ universal.  Taking $\rho \leq \lambda^{-1/2}$ and using
  \eqref{stimav}, we then get
  \begin{align}
    \int_{F \cap B_R(x)} {dy \over |x - y|} \leq C,
  \end{align}
  with a universal $C > 0$, for any $R \geq 1$ and $x \in \TTT_\ell$,
  provided that $\ell \gg 1$ independently of $x$.

  Recalling \eqref{crst} and arguing as in Lemma \ref{lem-ln}, for all
  $\ell \gg 1$ there exists $x_0\in \overline F_0$ such that
  \begin{align}
    C \ge \int_{F_0 \cap B_R(x_0)} {dy \over |x_0 - y|}  \geq c
    \min \{ \log \left( {\rm diam} \, F_0 \right), \log R \} \,,
  \end{align}
  for some universal $c, C > 0$. The claim then follows by choosing a
  universal $R$ that is sufficiently large.
\end{proof}

\section{Proof of Theorem \ref{thm-droplets}}\label{sec-drop}

For $\lambda > 0$, let $(u_\eps) \in \mathcal A_\eps$ be a family of
the regular representatives of minimizers of $E_\eps$, and let
$N_\eps$ and $u_{\eps,k} \in BV(\R^3; \{0, 1\})$ be as in the
statement of the theorem. Without loss of generality we may set
$x_{\eps,k} = 0$ in the statements below. We need to show that there
exists $\eps_0 > 0$ such that for all $\eps \leq \eps_0$:
\begin{enumerate}
\item[i)] There exist universal constants $C, c > 0$ such that
  \begin{align}
    \label{ccDueps-2}
    \| v_\eps \|_{L^\infty(\TTT_\ell)} \leq C, \qquad \quad
    \int_{\R^3} 
    u_{\eps,k}(x) \, dx \geq c \eps.
  \end{align}
\item[ii)] There exist universal constants $C, c > 0$ such that
  \begin{align} \label{cCNeps-2} %
    \mathrm{supp} (u_{\eps,k}) \subseteq B_{C \eps^{1/3}}(0), \qquad
    \quad c \lambda \eps^{-1/3} \leq N_\eps \leq C \lambda
    \eps^{-1/3}.
  \end{align}
\item[iii)] There exists a collection of indices $I_\eps$ such that
  $(\# I_\eps) / N_\eps \to 1$ as $\eps \to 0$ and, upon extraction of
  a subsequence, for every sequence $\eps_n \to 0$ and every
  $k_n \in I_{\eps_n}$ there holds $\tilde u_n \to \tilde u$ in
  $L^1(\R^3)$, where
  $\tilde u_n (x) := u_{\eps_n,k_n}(\eps_n^{1/3} x)$, and $\tilde u$
  is a minimizer of $\widetilde E_\infty$ over
  $\widetilde{\mathcal A}_\infty(m^*)$ for some
  $m^* \in \mathcal I^*$.
\end{enumerate}
The estimate for the potential in (i) follows from Theorem
\ref{lem-pot}, setting
$\tilde u_{\ell_\eps}=u_{\eps}(\cdot /\ell_\eps) \in \widetilde{\cal
  A}_{\ell_\eps}$
with $\ell_\eps = \eps^{-1/3}$ and noting that with
$\tilde u_{\ell_\eps} = \chi_F$ we have
$v_F = v_\eps(\cdot / \ell_\eps)$.  Similarly, the volume estimate in
(i) follows from Corollary \ref{cor-dens}. The inclusion in (ii)
follows from Lemma \ref{lem-diam} by a rescaling. The estimate for
$N_\eps$ in (ii) follows from (i) and the fact that
$\int_\TTT u_\eps \, dx=\lambda\eps^{2/3}$.

\medskip

We turn to the proof of statement (iii). Given $\delta>0$, let
$N_{\eps,\delta} \geq 0$ be the number of the components $u_{\eps,k}$
such that for $\tilde u_{\eps,k}(x) := u_{\eps,k}(\eps^{1/3} x)$,
we have
\begin{align}
  \widetilde E_\infty (\tilde u_{\eps,k}) \geq (f^*+\delta)
  \int_{\R^3} \tilde u_{\eps,k} \, dx. 
\end{align}
By \eqref{Eepsmin}, \eqref{cCNeps} and the arguments in the proof of
Proposition \ref{prp-lower} we have, as $\eps \to 0$,
\begin{align}
  \lambda f^* 
  &= \eps^{-4/3}E_\eps(u_\eps) + o_\eps(1) %
    \geq \eps^{1/3}\sum_{k=1}^{N_{\eps}}\widetilde
    E_\infty(\tilde u_{\eps,k}) + o_\eps(1) \notag \\
  &\geq \eps^{1/3}\Big(
    (f^*+\delta)\sum_{k=1}^{N_{\eps,\delta}} \int_{\R^3}
    \tilde u_{\eps,k} \, dx + f^* 
    \sum_{k=N_{\eps,\delta} + 1}^{N_{\eps}} \int_{\R^3}
    \tilde u_{\eps,k} \, dx \notag \Big) + o_\eps(1) \\ 
  &\geq \lambda f^* + c\,\delta
    N_{\eps,\delta}\,\eps^{1/3} + o_\eps(1), 
\end{align}
where we suitably ordered all $\tilde u_{\eps,k}$ and included a
possibility that the range of summation is empty in either of the two
sums. Hence, $N_{\eps,\delta} = o(\eps^{-1/3})$, and by (ii) it
follows that $N_{\eps,\delta} = o(N_\eps)$ for all $\delta>0$. This
implies that for every $\delta > 0$ there is $\eps_\delta > 0$ and a
collection of indices $I_{\eps_\delta}$ satisfying
$(\# I_{\eps_\delta}) / N_{\eps_\delta} \to 1$ such that
$\widetilde E_\infty (\tilde u_{{\eps_\delta},k}) / \int_{\R^3} \tilde
u_{{\eps_\delta},k} \, dx \to f^*$
uniformly in $k \in I_{\eps_\delta}$ as $\delta \to 0$. By (ii), for
every sequence of $\delta_n \to 0$ and every choice of
$k_n \in I_{\eps_{\delta_n}}$ the sequence
$\tilde u_n := \tilde u_{\eps_{\delta_n},k_n}$ is supported in
$B_R(0)$ for some $R > 0$ universal and equibounded in $BV(\R^3)$.
Hence, upon extraction of a subsequence we have
$\tilde u_n \to \tilde u$ in $L^1(\mathbb R^3)$ with
$m := \int_{\R^3} \tilde u \, dx > 0$. At the same time, by lower
semicontinuity of $\widetilde E_\infty$ we also have
$\widetilde E_\infty (\tilde u) / m \leq f^*$. Then, by Theorem
\ref{prp-fstar} the latter is, in fact, an equality, and so
$u_n(x) := \tilde u(\lambda_n x)$ with
$\lambda_n := (m^{-1} \int_{\R^3} \tilde u_{\eps,k} \, dx )^{1/3} \to
1$,
is a minimizing sequence for $\widetilde E_\infty$ over
$\widetilde {\mathcal A}_\infty(m)$ (cf. \eqref{Elam}). Thus,
$\tilde u$ is a minimizer of $\widetilde E_\infty$ over
$\widetilde {\mathcal A}_\infty(m)$. Again, by Theorem \ref{prp-fstar}
we then have $m \in \mathcal I^*$.  \qed

\paragraph*{Acknowledgements.}

The work of CBM was partly supported by NSF via grants DMS-0908279 and
DMS-1313687. The work of MN was partly supported by the Italian
CNR-GNAMPA and by the University of Pisa via grant PRA-2015-0017. CBM
gratefully acknowledges the hospitality of the University of
Heidelberg. HK and CBM gratefully acknowledge the hospitality of the
Max Planck Institute for Mathematics in the Sciences, and both CBM and
MN gratefully acknowledge the hospitality of Mittag-Leffler Institute,
where part of this work was completed.


\appendix

\section{Appendix}

We recall that by the Riesz-Fischer theorem, the space of signed Radon
measures $\MM(\TTT)$ is embedded in the space of distributions via the
identification
\begin{align} \label{app-h1dist} %
  \SKP{\phi}{\mu} := \int_\TTT \phi d\mu && \forall \phi \in
                                            C^\infty(\TTT).
\end{align}
On the other hand, any measure $\mu \in \MM^+(\TTT) \cap \mathcal H'$
(recall the definition in \eqref{def-H1*}) can be extended by
continuity to an element of the dual space $\mathcal H'$, which we
still denote by $\mu$, such that
\begin{align} \label{muriesz} \int_\TTT \phi\, d\mu =
  \SKPLL{\phi}{\mu}{{\mathcal H}}{{\mathcal H}'} && \forall
  \phi \in \mathcal H \cap C^0(\TTT).
\end{align}

\begin{lemma}\label{lemv}
  Let $\mu\in\MM^+(\TTT) \cap \mathcal H'$ and $u\in \mathcal H$.
  Then, up to taking the precise representative, $u$ belongs to
  $L^1(\TTT,d\mu)$ and
   
   \begin{align}\label{seitre}
     \SKPLL{u}{\mu}{{\mathcal H}}{{\mathcal H}'} = \int_\TTT u\, d\mu.
   \end{align}
 \end{lemma}

 \begin{proof}
   The result follows as in \cite[Theorem 1]{brezis79}. For the
   reader's convenience we include a simple alternative proof here.
   Since $u\in {\mathcal H}$, by \cite[Section 4.8: Theorem 1]{evans}
   we can identify $u$ with its precise representative and find a
   sequence $u_k \in \mathcal H\cap C^0(\TTT)$ such that $u_k \to u$
   in ${\mathcal H}$, and
   \begin{align}\label{eqN}
     u_k(x) \to u(x) &&\text{for all $x \not\in N$,}
   \end{align}
   where $N\subset\TTT$ is a set of zero inner capacity, that is, for
   any compact set $K\subset N$ there exists a sequence
   $\varphi_n\in \mathcal H\cap C^0(\TTT)$ such that $\varphi_n\to 0$
   in $\mathcal H$ and $\varphi_n=1$ on $K$.  Since
   $\mu\in \mathcal H'$ we have $\mu(K)=0$ for all compact
   $K\subset N$, so that
   \begin{align}\label{munot} 
   \mu(N)=\sup_{K\subset N}\mu(K)=0\,.
   \end{align}
   Since the functions $u_k$ are continuous for all $k \in \N$, we
   have
   \begin{align}
     \label{ukmu}
     \SKPLL{u_k}{\mu}{\mathcal H}{\mathcal H'} = \int_{\TTT} u_k d\mu, 
   \end{align}
   Therefore, by \eqref{muriesz} we get
   \begin{align}\label{eqalpha}
     \SKPLL{|u_k-u_{k'}|-\alpha_{k,k'}}{\mu}{\mathcal H}{\mathcal
     H'} = \int_{\TTT} |u_k-u_{k'}| d\mu -
     \alpha_{k,k'}\,\mu(\TTT),
   \end{align}
   for all $k' \in \N$, where
   \begin{align}
     \alpha_{k,k'}:= \int_\TTT |u_k-u_{k'}| \, dx\,.
   \end{align}
   It then follows
   \begin{eqnarray}\notag
     \| u_k-u_{k'}\|_{L^1(\TTT,d\mu)}
     &\leq& \|
            |u_k-u_{k'}|-\alpha_{k,k'}\|_{\mathcal H}\|\mu\|_{\mathcal H'} 
            +\mu(\TTT)\| u_k-u_{k'}\|_{L^1(\TTT)}
     \\\label{cauchy}
     &=& 
         \| \nabla (u_k-u_{k'}) 
         \|_{L^2(\TTT)}\|\mu\|_{\mathcal H'}+\mu(\TTT)\| 
         u_k-u_{k'}\|_{L^1(\TTT)}. 
   \end{eqnarray}
   Since $u_k$ is a Cauchy sequence in $\mathcal H$, hence also in
   $L^1(\TTT)$, from \eqref{cauchy} it follows that $u_k$ is a Cauchy
   sequence in $L^1(\TTT,d\mu)$ and, therefore, converges to some
   $\tilde u \in L^1(\TTT,d\mu)$. In fact, passing to a subsequence
   and using \eqref{eqN} and \eqref{munot}, we have
   $\tilde u(x) = u(x)$ for $\mu$-a.e.  $x \in \TTT$. Therefore, from
   \eqref{ukmu} we get
   \begin{align}
     \SKPLL{u}{\mu}{\mathcal H}{\mathcal H'} =\lim_{k\to\infty}
     \SKPLL{u_k}{\mu}{\mathcal H}{\mathcal H'} = \lim_{k\to\infty}
     \int_{\TTT} u_k \,d\mu = \int_{\TTT} u \,d\mu,
   \end{align}
   which concludes the proof.
\end{proof}

\medskip

The following lemma characterizes the measures in terms of the
Coulombic potential, see \cite[Lemma 3.2]{gms:arma13} for a related
result.
\begin{lemma}\label{lemhans}
  Let $\mu \in \MM^+(\TTT) \cap \mathcal H'$, and let
  $G:\TTT\to (-\infty,+\infty]$ be the unique distributional solution
  of \eqref{G} with $G(0) =+\infty$. Then the function
  \begin{align} \label{app-defv} %
    v(x) := \int_{\TTT}G(x-y)\,d\mu(y) \qquad x\in\TTT
  \end{align}
  belongs to $\mathcal H$ and solves
  \begin{align} \label{app-v} %
    - \int_\TTT v \Delta \varphi \, dx = \int_\TTT \varphi \, d \mu
    \qquad \qquad \forall \varphi \in C^\infty(\TTT) \cap \mathcal H.
  \end{align}
  Moreover $v\in L^1(\TTT,d\mu)$ and 
  \begin{align} \label{app-eq} %
    \int_ \TTT \int_\TTT G(x-y)\,d\mu(x) \, d\mu(y) = \int_{\TTT} v \,
    d\mu = \int_\TTT |\nabla v|^2 \, dx \,.
  \end{align}
\end{lemma}
\begin{proof}
  By the definition of $G$ and the fact that $G \in L^1(\TTT)$, the
  function $v$ belongs to $L^1(\TTT)$, solves \eqref{app-v} and has
  zero average on $\TTT$.  On the other hand, by \eqref{def-H1*} one
  can define a functional $T_\mu \in \mathcal H'$ such that
  $T_\mu(\varphi) = \int_\TTT \varphi \, d \mu$ for every $\varphi \in
  C^\infty(\TTT) \cap \mathcal H$.  Therefore, by Riesz Representation
  Theorem there exists $\tilde v \in \mathcal H$ such that
  \begin{align}\label{vweakH}
    - \int_\TTT v \Delta \varphi \, dx = \langle \varphi, \tilde v
    \rangle_{\mathcal H} = - \int_\TTT \tilde v \Delta \varphi \, dx
    \qquad \qquad \forall \varphi \in C^\infty(\TTT) \cap \mathcal H.
  \end{align}
  Thus, since $\Delta$ is a one-to-one map from $C^\infty(\TTT) \cap
  \mathcal H$ to itself, we conclude that $v = \tilde v$ almost
  everywhere with respect to the Lebesgue measure on $\TTT$ and,
  hence, $v \in \mathcal H$.

  Let now $\rho\in C^\infty(\TTT)$ be a radial symmetric-decreasing
  mollifier supported on $B_{1/8}(0)$, let $\rho_n(x):= n^3\rho(nx)$,
  so that $\rho_n \to \delta_0$ in $\mathcal D'(\TTT)$, and let
  $f_n \in C^\infty(\TTT)$ be defined as
  \begin{align}
    f_n(x):=\int_\TTT \rho_n(x-y)\,d\mu(y) \qquad x \in \TTT.
  \end{align}
  Then, if the measures $\mu_n\in\MM^+(\TTT) \cap \mathcal H'$ are
  such that $d\mu_n=f_n\,dx$, we have $T_{\mu_n} \to T_\mu$ in
  $\mathcal H'$ and $\mu_n \rightharpoonup \mu$ in $\MM(\TTT)$.
  Letting also $v_n(x):=\int_\TTT G(x-y)\,d\mu_n(y)$, we observe that
  $v_n\to v \in \mathcal H$, and
  $\mu_n\otimes \mu_n \rightharpoonup \mu\otimes \mu$ in
  $\MM(\TTT\times \TTT)$. For all $M>0$, we then get
  \begin{align}\label{nn}
    \int_ \TTT \int_\TTT G_M(x-y)\,d\mu(x)\, d\mu(y) = \lim_{n\to
    \infty}\int_\TTT \int_\TTT G_M(x-y)\,d\mu_n(x)\, d\mu_n(y)\,,
  \end{align}
  where we set $G_M(x):=\min(G(x),M)\in C(\TTT)$.  By Monotone
  Convergence Theorem we also have
  \begin{align}
    \int_\TTT \int_\TTT G(x-y)\,d\mu(x)\, d\mu(y) = \lim_{M\to \infty}
    \int_\TTT \int_\TTT G_M(x-y)\,d\mu(x) \, d\mu(y)\,.
  \end{align}
  Recalling \eqref{nn}, it then follows
  \begin{eqnarray}
    \int_\TTT \int_\TTT G(x-y)\,d\mu(x) \, d\mu(y) 
    &=& \lim_{M\to
        \infty}\lim_{n\to \infty}\int_\TTT \int_\TTT
        G_M(x-y)\,d\mu_n(x) \, d\mu_n(y)  \notag
    \\
    &\leq& \lim_{n\to \infty}\int_\TTT \int_\TTT
           G(x-y)\,d\mu_n(x) \, d\mu_n(y)  \notag
    \\ 
    &=& \lim_{n\to \infty}\int_{\TTT} v_n\, d\mu_n  = \lim_{n\to
        \infty} \|v_n\|^2_{\mathcal H} = \|v\|^2_{\mathcal H}\,.
  \end{eqnarray}
  Together with the fact that $G$ is bounded from below, by
  Fubini-Tonelli theorem this implies that $v\in L^1(\TTT,d\mu)$, with
  $\|v\|_{L^1(\TTT,d\mu)}\leq \|v\|^2_{\mathcal H}$.

  It remains to prove \eqref{app-eq}. We reason as in \cite[Theorem
  1.11]{landkof} and pass to the limit, as $n\to \infty$, in the
  equality
  \begin{align} \label{app-keq} %
    \int_\TTT v_n\, d\mu_n %
    = \int_\TTT |\nabla v_n|^2 \, dx\,,
  \end{align}
  which holds for all $n\in\mathbb N$.  Notice that the right-hand
  side of \eqref{app-keq} converges since $v_n \to v$ in $\mathcal H$,
  so that
  \begin{align} 
    \label{right} 
    \lim_{n\to \infty}\int_\TTT |\nabla v_n|^2 \, dx = \int_\TTT |\nabla
    v|^2 \, dx\,.
  \end{align}
  In order to pass to the limit in the left-hand side of
  \eqref{app-keq}, we write
  \begin{align}
    \int_\TTT v_n\, d\mu_n = \int_\TTT \int_\TTT
    G(x-y)\,d\mu_n(x)d\mu_n(y) = \int_\TTT \int_\TTT
    G_n(x-y)\,d\mu(x)d\mu(y)\,,
  \end{align}
  where we set
  \begin{eqnarray}
    G_n(x) &:=& \int_\TTT G(x-y)\tilde\rho_n(y)\,dy,
    \\
    \tilde\rho_n(x) &:=& \int_\TTT \rho_n(x-y)\rho_n(y)\,dy\,.
  \end{eqnarray} 
  We claim that there exists $C>0$ such that
  \begin{align} 
    \label{app-jeq}
    |G_n(x)| \leq C\left(1+|G(x)|\right)
  \end{align}
  for all $x\in\TTT$.  Indeed, we can write $G=\Gamma^\#+R$ as in
  \eqref{decompG}. Letting
  \begin{align}
    \Gamma^\#_n(x):=\int_\TTT \Gamma^\#
    (x-y)\tilde\rho_n(y)\,dy \qquad {\rm 
    and}\qquad R_n(x):=\int_\TTT R(x-y)\tilde\rho_n(y)\,dy\,,
  \end{align}
  we have that $R_n \to R$ uniformly as $n\to \infty$.  Moreover,
  since $\Gamma^\#$, $\Gamma^\#_n$ and $\tilde \rho_n$ are periodic
  when viewed as functions on $\mathbb R^3$, rewriting the integrals
  as integrals over subsets of $\mathbb R^3$ and applying Newton's
  Theorem we get
  \begin{align}
    \frac{\Gamma^\#_n(x)}{\Gamma^\#(x)} = 4 \pi |x|
    \int_{B_{1/4}(0)} \Gamma^\#(x - y) 
    \tilde\rho_n(y) \,dy = |x| \int_{B_{1/4}(0)} {\tilde\rho_n(y)
    \over |x - y|} \,dy = \int_{B_{|x|}(0)}\tilde\rho_n(y)\,dy
    \notag \\ + |x| \int_{B_{1/4}(0) \backslash B_{|x|}(0)}
    {\tilde\rho_n(y) \over |y|} \,dy \leq
    \int_{B_{1/4}(0)}\tilde\rho_n(y)\,dy = 1 \qquad {\rm for\ all\ }
    |x| < \frac14.
  \end{align}
  Since also
  \begin{align}
    \frac{\Gamma^\#_n(x)}{\Gamma^\#(x)} = 4 \pi |x|
    \int_{B_{1/8}(0)} \Gamma^\#(x - y) 
    \tilde\rho_n(y) \,dy \leq C \qquad {\rm for\ all\ } \frac14 \leq
    |x| \leq {\sqrt{3} \over 2},
  \end{align}
  this proves \eqref{app-jeq}.
	
  From the fact that $G_n(x)\to G(x)$ for all $x\in\TTT$, by
  \eqref{app-jeq} and the Dominated Convergence Theorem we get
  \begin{align} 
    \notag \lim_{n\to \infty}\int_\TTT v_n\,d\mu_n &= \lim_{n\to
      \infty}\int_\TTT \int_\TTT G_n(x-y)\,d\mu(x)d\mu(y)
    \\
    \label{left} &=\int_\TTT \int_\TTT G(x-y)\,d\mu(x)d\mu(y) =
    \int_\TTT v\, d\mu.
  \end{align}
  From \eqref{app-keq}, \eqref{right} and \eqref{left} we obtain
  \eqref{app-eq}.
 \end{proof}

 \begin{lemma} 
   \label{lemhans2} 
   Let $G$ be as in Lemma \ref{lemhans} and let $\mu \in \MM^+(\TTT)$
   satisfy \eqref{muCoul}. Then $\mu \in \MM^+(\TTT) \cap \mathcal
   H'$.
 \end{lemma}

 \begin{proof}
   Let $\varphi \in C^0(\TTT) \cap \mathcal H$. Using the same
   notation and arguments as in the proof of Lemma \ref{lemhans},
   with the help of Cauchy-Schwarz inequality we obtain
   \begin{align}
     \int_\TTT \varphi \, d \mu 
     & = \lim_{n \to \infty} \int_\TTT
       \varphi \, d \mu_n = \lim_{n \to \infty} \int_\TTT \nabla \varphi
       \cdot \nabla v_n \, dx \notag \\
     & \leq \| \varphi \|_{\mathcal H} \lim_{n \to \infty} \left(
       \int_\TTT \int_\TTT G(x
       - y) \, d \mu_n(x) \, d \mu_n(y) \right)^{\frac12} \notag \\
     & = \| \varphi \|_{\mathcal H} \lim_{n \to \infty} \left(
       \int_\TTT \int_\TTT G_n(x - y) \, d \mu(x) \, d \mu(y)
       \right)^{\frac12} \notag \\
     & = \| \varphi \|_{\mathcal H} \left( \int_\TTT \int_\TTT G(x -
       y) \, d \mu(x) \, d \mu(y) \right)^{\frac12},
   \end{align}
   which yields the inequality in \eqref{def-H1*}.
 \end{proof}

\bibliography{../nonlin,../mura,../stat}

\begin{thebibliography}{10}

\bibitem{gamow30}
G.~Gamow.
\newblock Mass defect curve and nuclear constitution.
\newblock {\em Proc. Roy. Soc. London A}, 126:632--644, 1930.

\bibitem{weizsacker35}
C.~F. von Weizs{\"a}cker.
\newblock Zur {Theorie} der {Kernmassen}.
\newblock {\em Zeitschrift f{\"u}r Physik A}, 96:431--458, 1935.

\bibitem{bohr36}
N.~Bohr.
\newblock Neutron capture and nuclear constitution.
\newblock {\em Nature}, 137:344--348, 1936.

\bibitem{bohr39}
N.~Bohr and J.~A. Wheeler.
\newblock The mechanism of nuclear fission.
\newblock {\em Phys. Rev.}, 56:426--450, 1939.

\bibitem{cohen62}
S.~Cohen and W.~J. Swiatecki.
\newblock The deformation energy of a charged drop: {IV. Evidence} for a
  discontinuity in the conventional family of saddle point shapes.
\newblock {\em Ann. Phys.}, 19:67--164, 1962.

\bibitem{myers66}
W.~D. Myers and W.~J. Swiatecki.
\newblock Nuclear masses and deformations.
\newblock {\em Nucl. Phys.}, 81:1--60, 1966.

\bibitem{cohen74}
S.~Cohen, F.~Plasil, and W.~J. Swiatecki.
\newblock Equilibrium configurations of rotating charged or gravitating liquid
  masses with surface tension. {II}.
\newblock {\em Ann. Phys.}, 82:557--596, 1974.

\bibitem{pelekasis90}
N.~A. Pelekasis, J.~A. Tsamopoulos, and G.~D. Manolis.
\newblock Equilibrium shapes and stability of charged and conducting drops.
\newblock {\em Phys. Fluids A: Fluid Dynamics}, 2:1328--1340, 1990.

\bibitem{myers96}
W.~D. Myers and W.~J. Swiatecki.
\newblock Nuclear properties according to the {Thomas-Fermi} model.
\newblock {\em Nucl. Phys. A}, 601:141--167, 1996.

\bibitem{cook}
N.~D. Cook.
\newblock {\em Models of the Atomic Nucleus}.
\newblock Springer, Berlin, 2006.

\bibitem{meitner39}
L.~Meitner and O.~R. Frisch.
\newblock Disintegration of uranium by neutrons: a new type of nuclear
  reaction.
\newblock {\em Nature}, 143:239--240, 1939.

\bibitem{feenberg39}
E.~Feenberg.
\newblock On the shape and stability of heavy nuclei.
\newblock {\em Phys. Rev.}, 55:504--505, 1939.

\bibitem{frenkel39}
J.~Frenkel.
\newblock On the splitting of heavy nuclei by slow neutrons.
\newblock {\em Phys. Rev.}, 55:987--987, 1939.

\bibitem{baym71}
G.~Baym, H.~A. Bethe, and C.~J. Pethick.
\newblock Neutron star matter.
\newblock {\em Nucl. Phys. A}, 175:225--271, 1971.

\bibitem{koester90}
D.~Koester and G.~Chanmugam.
\newblock Physics of white dwarf stars.
\newblock {\em Rep. Prog. Phys.}, 53:837--915, 1990.

\bibitem{pethick95}
C.~J. Pethick and D.~G. Ravenhall.
\newblock Matter at large neutron excess and the physics of neutron-star
  crusts.
\newblock {\em Ann. Rev. Nucl. Part. Sci.}, 45:429--484, 1995.

\bibitem{kirzhnits60}
D.~A. Kirzhnits.
\newblock Internal structure of super-dense stars.
\newblock {\em Sov. Phys. -- JETP}, 11:365--368, 1960.

\bibitem{abrikosov61}
A.~A. Abrikosov.
\newblock Some properties of strongly compressed matter. {I}.
\newblock {\em Sov. Phys. -- JETP}, 12:1254--1259, 1961.

\bibitem{salpeter61}
E.~E. Salpeter.
\newblock Energy and pressure of a zero-temperature plasma.
\newblock {\em Astrophys. J.}, 134:669--682, 1961.

\bibitem{ravenhall83}
D.~G. Ravenhall, C.~J. Pethick, and J.~R. Wilson.
\newblock Structure of matter below nuclear saturation density.
\newblock {\em Phys. Rev. Lett.}, 50:2066--2069, 1983.

\bibitem{hashimoto84}
M.~Hashimoto, H.~Seki, and M.~Yamada.
\newblock Shape of nuclei in the crust of neutron star.
\newblock {\em Prog. Theor. Phys.}, 71:320--326, 1984.

\bibitem{oyamatsu84}
K.~Oyamatsu, M.~Hashimoto, and M.~Yamada.
\newblock Further study of the nuclear shape in high-density matter.
\newblock {\em Prog. Theor. Phys.}, 72:373--375, 1984.

\bibitem{lattimer85}
J.~M. Lattimer, C.~J. Pethick, D.~G. Ravenhall, and D.~Q. Lamb.
\newblock Physical properties of hot, dense matter: {The} general case.
\newblock {\em Nucl. Phys. A}, 432:646--742, 1985.

\bibitem{lorenz93}
C.~P. Lorenz, D.~G. Ravenhall, and C.~J. Pethick.
\newblock Neutron star crusts.
\newblock {\em Phys. Rev. Lett.}, 70:379--382, 1993.

\bibitem{okamoto13}
M.~Okamoto, T.~Maruyama, K.~Yabana, and T.~Tatsumi.
\newblock Nuclear ``pasta'' structures in low-density nuclear matter and
  properties of the neutron-star crust.
\newblock {\em Phys. Rev. C}, 88:025801, 2013.

\bibitem{schneider13}
A.~Schneider, C.~Horowitz, J.~Hughto, and D.~Berry.
\newblock Nuclear ``pasta'' formation.
\newblock {\em Phys. Rev. C}, 88:065807, 2013.

\bibitem{ohta86}
T.~Ohta and K.~Kawasaki.
\newblock Equilibrium morphologies of block copolymer melts.
\newblock {\em Macromolecules}, 19:2621--2632, 1986.

\bibitem{ren00}
X.~F. Ren and J.~C. Wei.
\newblock On the multiplicity of solutions of two nonlocal variational
  problems.
\newblock {\em SIAM J. Math. Anal.}, 31:909--924, 2000.

\bibitem{choksi03}
R.~Choksi and X.~Ren.
\newblock On the derivation of a density functional theory for microphase
  separation of diblock copolymers.
\newblock {\em J. Statist. Phys.}, 113:151--176, 2003.

\bibitem{m:pre02}
C.~B. Muratov.
\newblock Theory of domain patterns in systems with long-range interactions of
  {Coulomb} type.
\newblock {\em Phys. Rev. E}, 66:066108 pp. 1--25, 2002.

\bibitem{choksi01}
R.~Choksi.
\newblock Scaling laws in microphase separation of diblock copolymers.
\newblock {\em J. Nonlinear Sci.}, 11:223--236, 2001.

\bibitem{alberti09}
G.~Alberti, R.~Choksi, and F.~Otto.
\newblock Uniform energy distribution for an isoperimetric problem with
  long-range interactions.
\newblock {\em J. Amer. Math. Soc.}, 22:569--605, 2009.

\bibitem{m:cmp10}
C.~B. Muratov.
\newblock Droplet phases in non-local {Ginzburg-Landau} models with {Coulomb}
  repulsion in two dimensions.
\newblock {\em Comm. Math. Phys.}, 299:45--87, 2010.

\bibitem{choksi10}
R.~Choksi and M.~A. Peletier.
\newblock Small volume fraction limit of the diblock copolymer problem: {I.
  Sharp} interface functional.
\newblock {\em SIAM J. Math. Anal.}, 42:1334--1370, 2010.

\bibitem{choksi11}
R.~Choksi and M.~A. Peletier.
\newblock Small volume fraction limit of the diblock copolymer problem: {II.
  Diffuse} interface functional.
\newblock {\em SIAM J. Math. Anal.}, 43:739--763, 2011.

\bibitem{km:cpam12}
H.~Kn\"upfer and C.~B. Muratov.
\newblock On an isoperimetric problem with a competing non-local term. {I. The}
  planar case.
\newblock {\em Comm. Pure Appl. Math.}, 66:1129--1162, 2013.

\bibitem{km:cpam13}
H.~Kn\"upfer and C.~B. Muratov.
\newblock On an isoperimetric problem with a competing non-local term. {II.
  The} general case.
\newblock {\em Commun. Pure Appl. Math.}, 67:1974--1994, 2014.

\bibitem{gms:arma13}
D.~Goldman, C.~B. Muratov, and S.~Serfaty.
\newblock The {$\Gamma$}-limit of the two-dimensional {Ohta-Kawasaki} energy.
  {I. Droplet} density.
\newblock {\em Arch. Rational Mech. Anal.}, 210:581--613, 2013.

\bibitem{gms:arma14}
D.~Goldman, C.~B. Muratov, and S.~Serfaty.
\newblock The {$\Gamma$}-limit of the two-dimensional {Ohta-Kawasaki} energy.
  {Droplet} arrangement via the renormalized energy.
\newblock {\em Arch. Rational Mech. Anal.}, 212:445--501, 2014.

\bibitem{cicalese13}
M.~Cicalese and E.~Spadaro.
\newblock Droplet minimizers of an isoperimetric problem with long-range
  interactions.
\newblock {\em Comm. Pure Appl. Math.}, 66:1298--1333, 2013.

\bibitem{acerbi13}
E.~Acerbi, N.~Fusco, and M.~Morini.
\newblock Minimality via second variation for a nonlocal isoperimetric problem.
\newblock {\em Commun. Math. Phys.}, 322:515--557, 2013.

\bibitem{julin14}
V.~Julin.
\newblock Isoperimetric problem with a {Coulombic} repulsive term.
\newblock {\em Indiana Univ. Math. J.}, 63:77--89, 2014.

\bibitem{julin13}
V.~Julin and G.~Pisante.
\newblock Minimality via second variation for microphase separation of diblock
  copolymer melts.
\newblock Preprint: arXiv:1301.7213, 2013.

\bibitem{bonacini14}
M.~Bonacini and R.~Cristoferi.
\newblock Local and global minimality results for a nonlocal isoperimetric
  problem on {$\mathbb R^N$}.
\newblock {\em SIAM J. Math. Anal.}, 46:2310--2349, 2014.

\bibitem{lu14}
J.~Lu and F.~Otto.
\newblock Nonexistence of minimizer for {Thomas-Fermi-Dirac-von Weizs\"acker}
  model.
\newblock {\em Comm. Pure Appl. Math.}, 67:1605--1617, 2014.

\bibitem{figalli15}
A.~Figalli, N.~Fusco, F.~Maggi, V.~Millot, and M.~Morini.
\newblock Isoperimetry and stability properties of balls with respect to
  nonlocal energies.
\newblock {\em Commun. Math. Phys.}, 336:441--507, 2015.

\bibitem{mz:agag15}
C.~B. Muratov and A.~Zaleski.
\newblock On an isoperimetric problem with a competing non-local term:
  quantitative results.
\newblock {\em Ann. Glob. Anal. Geom.}, 47:63--80, 2015.

\bibitem{choksi12}
R.~Choksi.
\newblock On global minimizers for a variational problem with long-range
  interactions.
\newblock {\em Quart. Appl. Math.}, 70:517--537, 2012.

\bibitem{frank15}
R.~L. Frank and E.~H. Lieb.
\newblock A compactness lemma and its application to the existence of
  minimizers for the liquid drop model.
\newblock arXiv:1503.00192, 2015.

\bibitem{sandier12}
E.~Sandier and S.~Serfaty.
\newblock From the {Ginbzurg-Landau} model to vortex lattice problems.
\newblock {\em Comm. Math. Phys.}, 313:635--743, 2012.

\bibitem{rougerie13}
N.~Rougerie and S.~Serfaty.
\newblock Higher dimensional {Coulomb} gases and renormalized energy
  functionals.
\newblock Preprint: arXiv:1307.2805, 2013.

\bibitem{tinkham}
M.~Tinkham.
\newblock {\em Introduction to superconductivity}.
\newblock McGraw-Hill, New York, 2nd edition, 1996.

\bibitem{chen07arma}
X.~Chen and Y.~Oshita.
\newblock An application of the modular function in nonlocal variational
  problems.
\newblock {\em Arch. Ration. Mech. Anal.}, 186:109--132, 2007.

\bibitem{fuchs35}
K.~Fuchs.
\newblock A quantum mechanical investigation of the cohesive forces of metallic
  copper.
\newblock {\em Proc. Roy. Soc. London A}, 151:585--602, 1935.

\bibitem{foldy71}
L.~L. Foldy.
\newblock Phase transition in a {Wigner} lattice.
\newblock {\em Phys. Rev. B}, 3:3472--3479, 1971.

\bibitem{nagai83}
T.~Nagai and H.~Fukuyama.
\newblock Ground state of a {Wigner} crystal in a magnetic field. {II.
  Hexagonal} close-packed structure.
\newblock {\em J. Phys. Soc. Japan}, 52:44--53, 1983.

\bibitem{dobrynin05}
A.~V. Dobrynin and M.~Rubinstein.
\newblock Theory of polyelectrolytes in solutions and at surfaces.
\newblock {\em Progr. Polym. Sci.}, 30:1049--1118, 2005.

\bibitem{forster04}
S.~F\"orster, V.~Abetz, and A.~H.~E. M\"uller.
\newblock Polyelectrolyte block copolymer micelles.
\newblock {\em Adv. Polym. Sci.}, 166:173--210, 2004.

\bibitem{sternberg11}
P.~Sternberg and I.~Topaloglu.
\newblock On the global minimizers of the nonlocal isoperimetric problem in two
  dimensions.
\newblock {\em Interfaces Free Bound.}, 13:155--169, 2010.

\bibitem{fusco08}
N.~Fusco, F.~Maggi, and A.~Pratelli.
\newblock The sharp quantitative isoperimetric inequality.
\newblock {\em Ann. of Math.}, 168:941--980, 2008.

\bibitem{rigot00}
S.~Rigot.
\newblock Ensembles quasi-minimaux avec contrainte de volume et
  rectifiabilit\'e uniforme.
\newblock {\em M\'emoires de la SMF, 2e s\'erie}, 82:1--104, 2000.

\bibitem{gilbarg}
D.~Gilbarg and N.~S. Trudinger.
\newblock {\em Elliptic Partial Differential Equations of Second Order}.
\newblock Springer-Verlag, Berlin, 1983.

\bibitem{choksi07}
R.~Choksi and P.~Sternberg.
\newblock On the first and second variations of a nonlocal isoperimetric
  problem.
\newblock {\em J. Reine Angew. Math.}, 611:75--108, 2007.

\bibitem{maggi}
F.~Maggi.
\newblock {\em Sets of Finite Perimeter and Geometric Variational Problems}.
\newblock Cambridge Studies in Advanced Mathematics, 135. Cambridge University
  Press, Cambridge, 2012.

\bibitem{goldman12cvar}
M.~Goldman and M.~Novaga.
\newblock Volume-constrained minimizers for the prescribed curvature problem in
  periodic media.
\newblock {\em Calc. Var. PDE}, 44:297--318, 2012.

\bibitem{brezis79}
H.~Brezis and F.~Browder.
\newblock A property of {S}obolev spaces.
\newblock {\em Comm. Partial Differential Equations}, 4:1077--1083, 1979.

\bibitem{evans}
L.~C. Evans and R.~L. Gariepy.
\newblock {\em Measure Theory and Fine Properties of Functions}.
\newblock CRC, Boca Raton, 1992.

\bibitem{landkof}
N.~S. Landkof.
\newblock {\em Foundations of modern potential theory}.
\newblock Springer-Verlag, New York, 1972.
\newblock Translated from the Russian by A. P. Doohovskoy, Die Grundlehren der
  mathematischen Wissenschaften, Band 180.

\end{thebibliography}
\bibliographystyle{unsrt}

\end{document}